\newif\ifpublish
\publishtrue

\documentclass[runningheads]{llncs}
\usepackage{caption}

\usepackage[cmex10]{amsmath}
\usepackage{amsfonts}
\usepackage{algorithm}
\usepackage[noend]{algpseudocode}
\usepackage{graphicx}
\usepackage{textcomp}
\usepackage{xcolor}
\usepackage{xspace}
\usepackage{todonotes}
\usepackage{pifont}
\usepackage{tikz}
\usepackage{cite}
\usepackage{subcaption}
\usepackage[font={footnotesize}]{caption}
\usepackage[hyphens]{url}
\usepackage[hidelinks]{hyperref}
\usepackage{cleveref}
\usepackage{booktabs}
\usepackage{pifont}
\newcommand{\cmark}{\ding{51}}%
\newcommand{\xmark}{\ding{55}}%
\usepackage{makecell}

\usepackage{enumitem}
\setitemize{noitemsep,topsep=0pt,parsep=0pt,partopsep=0pt}

\ifdefined\publish
    \newcommand{\alberto}[1]{}
    \newcommand{\dm}[1]{}
    \newcommand{\dmsuggest}[1]{}
    \newcommand{\lef}[1]{}
    \newcommand{\george}[1]{}
    
\else
    \newcommand{\alberto}[1]{\todo[inline, color=yellow!40]{\textbf{Alberto:} #1}}
    \newcommand{\dm}[1]{\todo[inline, color=green!40]{\textbf{Dahlia:} #1}}
    \newcommand{\dmsuggest}[1]{\todo[inline, color=magenta!40]{\textbf{DM suggestion:} #1}}
    \newcommand{\lef}[1]{\todo[inline, color=orange!40]{\textbf{Lefteris:} #1}}
    \newcommand{\george}[1]{\todo[inline, color=blue!40]{\textbf{George:} #1}}

\fi

\newcommand{\para}[1]{\smallskip\noindent\textbf{#1.}}

\newtheorem{assumption}{Assumption}
\newcommand{\algfontsize}{\scriptsize}

\newcommand{\sysname}{Pilotfish\xspace}
\newcommand{\sys}{\sysname}

\newcommand{\primary}{Primary\xspace}

\newcommand{\seqworker}{SequencingWorker\xspace}
\newcommand{\execworker}{ExecutionWorker\xspace}
\newcommand{\seqworkers}{SequencingWorkers\xspace}
\newcommand{\execworkers}{ExecutionWorkers\xspace}
\newcommand{\ew}[2]{\ensuremath{\textit{EW}_{#1,#2}}}
\newcommand{\peers}{\textbf{peers}\xspace}
\newcommand{\readto}{\textbf{read-to}\xspace}
\newcommand{\readfrom}{\textbf{read-from}\xspace}
\newcommand{\stabletxid}{\textbf{stable-txid}\xspace}

\newcommand{\self}{\text{\textbf{Self}}\xspace}
\newcommand{\none}{\ensuremath{\perp}}

\newcommand{\delete}{\text{\textbf{Delete}}\xspace}

\newcommand{\R}{\ensuremath{\mathcal{R}}}
\newcommand{\W}{\ensuremath{\mathcal{W}}}
\newcommand{\Read}{\ensuremath{\textsf{R}}}

\newcommand{\Write}{\ensuremath{\textsf{W}}}

\newcommand{\object}{\ensuremath{o}\xspace}

\newcommand{\objectid}{\ensuremath{oid}\xspace}
\newcommand{\objectversion}{ \textsf{Version}}
\newcommand{\tx}{\textsf{Tx}\xspace}
\newcommand{\txidx}{\textsf{TxIdx}\xspace}
\newcommand{\batch}{\textsf{Batch}\xspace}
\newcommand{\batchid}{\textsf{BatchId}\xspace}
\newcommand{\batchidx}{\textsf{BatchIdx}\xspace}
\newcommand{\txaugmented}{\textsf{Tx+}\xspace}

\newcommand{\executedIdx}{\textbf{j}\xspace}
\newcommand{\executed}{\textbf{E}\xspace}
\newcommand{\loadedIdx}{\textbf{i}\xspace}
\newcommand{\loaded}{\textbf{B}\xspace}
\newcommand{\objectsreceived}{\textbf{R}\xspace}
\newcommand{\checkpoints}{\textbf{C}\xspace}
\newcommand{\sentmessages}{\textbf{S}\xspace}
\newcommand{\checkpointsIdx}{\textbf{i}\xspace}
\newcommand{\checkpointsnotification}{\textbf{R}\xspace}

\newcommand{\propose}{\textsf{ProposeMessage}\xspace}
\newcommand{\ready}{\textsf{ReadyMessage}\xspace}
\newcommand{\result}{\textsf{ResultMessage}\xspace}
\newcommand{\updatepropose}{\textsf{UpdateProposeExec}\xspace}
\newcommand{\checkpointed}{\textsf{CheckpointedMessage}\xspace}
\newcommand{\recover}{\textsf{Recover}\xspace}
\newcommand{\recoverok}{\textsf{RecoverOk}\xspace}
\newcommand{\notifysync}{\textsf{NotifySync}\xspace}
\newcommand{\sync}{\textsf{Sync}\xspace}

\newcommand{\synccomplete}{\textsf{SyncComplete}\xspace}
\newcommand{\newreader}{\textsf{NewReader}\xspace}
\newcommand{\newreaderok}{\textsf{NewReaderOk}\xspace}

\newcommand{\recoverabort}{\textsf{Abort}\xspace}

\newcommand{\augmentedtx}{\textsf{AugTx}\xspace}

\newcommand{\batchdb}{\textsc{Batches}\xspace}
\newcommand{\objectsdb}{\textsc{Objects}\xspace}
\newcommand{\pendingdb}{\textsc{Pending}\xspace}
\newcommand{\txdb}{\textsc{ReceivedObj}\xspace}
\newcommand{\missingdb}{\textsc{Missing}\xspace}

\newcommand{\readset}[1]{\R(#1)\xspace}

\newcommand{\writeset}[1]{\W(#1)\xspace}
\newcommand{\exec}[1]{\emph{exec}(#1)\xspace}

\newcommand{\len}[1]{\emph{len}(#1)\xspace}

\newcommand{\one}{\ding{202}\xspace}
\newcommand{\two}{\ding{203}\xspace}
\newcommand{\three}{\ding{204}\xspace}
\newcommand{\four}{\ding{205}\xspace}
\newcommand{\five}{\ding{206}\xspace}

\newcommand{\set}[1]{\ensuremath{\{#1\}}} 

\let\llncssubparagraph\subparagraph
\let\subparagraph\paragraph
\usepackage[compact]{titlesec}
\titlespacing*{\section}{0pt}{*1}{*1}
\titlespacing*{\subsection}{0pt}{*1}{*1}
\titlespacing*{\subsubsection}{0pt}{*1}{*1}

\let\subparagraph\llncssubparagraph

\selectfont

\begin{document}
\title{
    \sysname: Distributed Execution for\\Scalable Blockchains
}

\author{
    Quentin Kniep\inst{1}\and
    Lefteris Kokoris-Kogias\inst{2,3}\and
    Alberto Sonnino\inst{3,4}\and
    Igor~Zablotchi\inst{3}\and
    Nuda Zhang\inst{5}    
}
\authorrunning{Q. Kniep, L. Kokoris-Kogias, A. Sonnino, I. Zablotchi, N. Zhang}

\institute{ETH Zurich \email{qkniep@ethz.ch}\and
IST Austria\and
Mysten Labs \email{\{lefteris,alberto,igor\}@mystenlabs.com}\and
University College London (UCL)\and
University of Michigan \email{nudzhang@umich.edu}
}



\maketitle
\begin{abstract}

Scalability is a crucial requirement for modern large-scale systems, enabling elasticity and ensuring responsiveness under varying load. While cloud systems have achieved scalable architectures, blockchain systems remain constrained by the need to over-provision validator machines to handle peak load. This leads to resource inefficiency, poor cost scaling, and limits on performance. To address these challenges, we introduce \sysname, the first scale-out transaction execution engine for blockchains. \sysname enables validators to scale horizontally by distributing transaction execution across multiple worker machines, allowing elasticity without compromising consistency or determinism. It integrates seamlessly with the lazy blockchain architecture, completing the missing piece of execution elasticity. To achieve this, \sysname tackles several key challenges: ensuring scalable and strongly consistent distributed transactions, handling partial crash recovery with lightweight replication, and maintaining concurrency with a novel versioned-queue scheduling algorithm.
Our evaluation shows that \sysname scales linearly up to at least eight workers per validator for compute-bound workloads, while maintaining low latency. By solving scalable execution, \sysname brings blockchains closer to achieving end-to-end elasticity, unlocking new possibilities for efficient and adaptable blockchain systems.

\end{abstract}

\section{Introduction}

A crucial property required by modern large-scale computing is \textit{scalability}~\cite{scalability}, which refers to a system's ability to dynamically adapt its performance as load changes, ensuring that the system remains responsive despite varying load. Scalability is fundamental because it is an essential requirement for elasticity~\cite{elasticity}, and thus in turn for a good user experience (e.g., responsiveness) at a sustainable cost. Without elasticity, systems either risk being overwhelmed during peak loads, leading to poor performance and user dissatisfaction, or they incur excessive costs during low-load periods by maintaining unnecessary resources.

Over the past decades, significant effort has been devoted to developing scalable software architectures for cloud-based systems~\cite{elasticitySOTA}. 
However, the situation is starkly different for blockchain systems. Among core blockchain tasks, transaction execution is particularly challenging with respect to scalability. The current dominant approach to transaction execution in blockchain involves ensuring that validator machines are sufficiently powerful to handle peak loads~\cite{solanareqs,suireqs,aptosreqs}. This approach is scalable up to a point, but has limitations: (1)~it leads to resource inefficiency, as validators remain over-provisioned during low-load periods; (2)~it has a resource ceiling, as even the most powerful single machine will eventually be insufficient if the load is high enough; (3)~it has poor cost scaling, as high-end machines are expensive and limited to a few vendors.

In response to these challenges, we introduce \sysname, the first scale-out transaction execution engine for blockchain. The core idea of \sysname is to run each validator on multiple mutually trusting machines or workers, as opposed to running a single machine per validator. Each worker is only responsible for a subset of the validator’s state, and only executes a subset of transactions. This approach opens the way toward elasticity, as it allows scaling each validator out and in as the load increases and decreases. 

\sysname is designed to integrate seamlessly with the lazy blockchain architecture, which is increasingly used by modern blockchains~\cite{celestia, lazyledger, prism,stefoexecuting, narwhal, bullshark, snap-and-chat, ebb-and-flow, mysticeti, dumbo-ng}: as of the time of writing, lazy blockchains account for over \$$20$ billion in market capitalization~\cite{coinmarketcap}.%
\ifpublish%
\fi 
Lazy blockchains separate the problems of transaction dissemination, ordering, and execution. They provide a scalable solution to two of the three core blockchain tasks: dissemination (ensuring that client transactions are available at a quorum of validators) and ordering (establishing a reliable total order over transactions, also known as consensus). However, as mentioned above, state-of-the-art lazy blockchains do not solve the \textit{execution scalability} problem: their execution is still designed to run on a single machine. 

\sysname must address several challenges to achieve this. First (i), it must solve the distributed transaction problem, since the validator state is sharded across multiple worker machines, and transactions may span multiple shards. This is especially challenging since blockchains need to guarantee strong consistency (serializability) and determinism, without compromising on latency or throughput.
Most existing approaches to distributed transactions cannot directly be applied to our setting: (1)
the two-phase commit approach~\cite{two-phase} guarantees strong consistency but is not scalable; (2) the relaxed consistency approach~\cite{cassandra, pnuts} is scalable but sacrifices strong consistency, which is crucial for blockchain; (3) the restricted transaction approach~\cite{elastras, sinfonia} is both scalable and strongly consistent, but sacrifices transaction generality. The most promising existing solution for our needs is that of deterministic databases~\cite{calvin}, which balance scalability, consistency, and transaction generality. We borrow techniques from distributed databases and leverage the fact that in lazy blockchains, consensus precedes—and is decoupled from—execution, so by execution time, validators have agreed on a permanent ordering of transactions. 


Secondly (ii), \sysname needs to tolerate workers crashing and recovering. 
To address this, \sysname maintains sufficient state among workers as \emph{checkpoints} to allow recovering machines to catch up with the rest. A straightfoward solution would be to resort to strong (and expensive) consensus-based replication techniques among workers internal to the validator~\cite{calvin,lamport2001paxos,ongaro2015raft}. However \sysname avoids such overhead by observing that consistency and availability of the commit sequence are already provided by the blockchain protocol. Thus, \sysname optimistically does lightweight, best-effort replication between workers, and relies on recovery from other validators only if optimistic replication fails.

Finally (iii), \sysname aims to support a simpler programming model where transactions may only partially specify their input read and write set (e.g., as required for Move~\cite{sui-move}). This, however, creates an additional challenge for \sysname, as objects that might be accessed dynamically at execution time can be located in different workers. This means that objects cannot be overwritten until all previous transactions have finished, effectively reverting to sequential execution and enforcing write-after-write dependencies. This limitation would reduce the parallelizability of the workload. 
\sysname circumvents this issue by leveraging its enforced determinism, allowing in-memory execution to be lost and safely recovered in the event of crashes. \sysname relies on a novel versioned-queue scheduling algorithm that allows transactions with write-after-write conflicts to execute concurrently. We couple this with our crash recovery mechanism, which only persists consistent states. As a result, upon a crash, \sysname simply re-executes a few transactions, but thanks to the deterministic nature of the blockchain this does not pose any inconsistency risks.

We evaluate \sysname by studying its latency and throughput, while varying the number of workers per validator, the computational intensity, and the degree of contention of the workload.
We find that \sysname scales linearly to at least $8$ workers per validator when the workload is compute-bound, while keeping latency under $50$ ms.

\para{Discussion}
While this work focuses on a scalable protocol for distributed blockchain transactions, achieving full elasticity poses additional challenges, particularly in dynamically scaling up and down workers and repartitioning objects. These aspects, while critical to practical implementations, are well-explored in existing literature on elastic systems: dynamic workload partitioning\cite{accordion, dynamo}, load-aware worker scaling\cite{aws-autoscaling, accordion}, and online object migration~\cite{accordion, pnuts}.


\section{System Model}
\label{sec:model}

\sysname implements a blockchain \textit{validator}, composed internally of a black-box \textit{Primary} machine, as well as a set of worker machines, simply called \textit{workers}. 
The Primary is responsible for communicating with (the Primaries of) other validators in order to agree on an ordered sequence of transactions. The workers collectively execute the ordered sequence of transactions and update the validator's state accordingly. 

\para{Objects and Transactions} \sysname validators replicate the state of the blockchain represented as a set of \emph{objects}~\cite{objectmodel}. Transactions can read and write (mutate, create, and delete) objects, and reference every object by its unique identifier $\objectid$.
A transaction is an authenticated command that references a set of objects (by their unique identifier $\objectid$), and an entry function into a smart contract call identifying the execution code. The transaction divides the objects it references into two disjoint sets, (i) the read set $\R$ referencing input objects that the transaction may only read, and (ii) the write set $\W$ referencing objects that the transaction may mutate.
In most cases, the identifier $\objectid$ of each object of the read and write sets can be computed using only the information provided by the transaction, without the need to execute it or access any object's data. In these cases, \sysname has complete knowledge of the read and write sets of the transaction. However, \sysname also supports dynamic accesses (\Cref{sec:dynamic}) where the read and write set of a transaction is discovered only upon attempting to execute the transaction, adopting the execution model of Sui~\cite{sui-lutris}.

\para{Network Model}
We assume that the Primary and workers communicate by sending messages over the network through point-to-point connections. 
We assume that the network is fully connected and reliable: each message sent by a correct process (i.e., non-faulty machine) to a correct process is eventually delivered. Furthermore, we assume authenticated channels: the receiver of a message is aware of the sender's identity.

\para{Synchrony Model} We consider the standard partially synchronous environment~\cite{dwork1988consensus}.
Specifically, there exists an unknown Global Stabilization Time (GST) and a positive known duration $\delta$ such that message delays are bounded by $\delta$ after GST: a message sent at time $\tau$ is received by time $\max(\tau, \text{GST}) + \delta$.
It has been shown that in partial synchrony, crash failures can eventually be perfectly detected~\cite{chandra1996, crash-recovery-fd}, thus we assume an eventually perfect failure detector.

\para{Threat Model} \label{sec:threat-model}
We assume that each validator is controlled by a single entity, or by a set of mutually trusting entities. This implies that the Primary and workers trust each other, and we therefore only consider crash failures for components internal to the validator (a validator as a whole may still exhibit Byzantine behavior in its interaction with other validators, but tolerating such failures is handled by the blockchain protocol, which is outside the scope of this work). For this reason, we do not require any cryptography assumptions, other than the authenticated channels.\footnote{
Our network, synchrony, trust and cryptography assumptions only apply \textit{internally} to the validator. By contrast, the outer blockchain protocol, which governs how validators interact with each other, may make entirely different assumptions on synchrony and types of failures and thus may require stronger cryptography primitives.} 
For pedagogical reasons, for the first part of the paper we assume that workers cannot fail. Later in \Cref{sec:cft-short}, we expand each logical worker to have a set of $n_e=2f_e+1$ replicas such that as long as for each worker there are $f_e+1$ replicas available the system remains live and safe.\footnote{Here, $f_e$ refers to the number of replicas that may crash per logical worker, \textit{internally to the validator}; in particular, $f_e$ may be different from the number of validators in the blockchain that may be Byzantine (usually denoted by $f$).} 
In case this threshold is breached the validator can still synchronize with the rest of the validators of the lazy blockchain through a standard recovery procedure~\cite{narwhal} that is out of scope.

\para{Core Properties} \label{sec:properties}
\ifpublish
\Cref{sec:proofs}
\fi 
proves that \sysname guarantees serializability, determinism, and liveness. Intuitively, serializability means that \sysname execution produces the same result as a sequential execution. Determinism means that every correct validator receiving the same sequence of transactions performs the same state transitions. Liveness means that all correct validators receiving a sequence of transactions eventually execute it.

\begin{definition}[\sysname Serializability] \label{def:serializability}
    A correct validator executing the sequence of transactions $[\tx_1, \dots, \tx_n]$ holds the same state as if the transactions were executed sequentially, in the given order.
\end{definition}

\begin{definition}[\sysname Determinism] \label{def:determinism}
    No two correct validators that executed the same sequence of transactions $[\tx_1, \dots, \tx_n]$ have different states.
\end{definition}

\begin{definition}[\sysname Liveness] \label{def:liveness}
    Correct validators receiving the sequence of transactions $[\tx_1, \dots, \tx_n]$ eventually execute all transactions $\tx_1, \dots, \tx_n$.
\end{definition}

\section{Existing Designs \& \sysname Overview} 
\label{sec:overview}

\subsection{Previous Designs}

Previous designs for scaling execution in lazy blockchains fall into two categories. The first is parallel execution~\cite{solana-vm,gelashvili2023blockstm,sui-lutris}, where each validator uses a high-end server to handle increased load. This approach lacks elasticity: the cost of running a powerful validator remains high regardless of actual load, leading to inefficiency during low usage and performance ceilings due to finite server resources.

The second category employs inter-validator sharding~\cite{omniledger,chainspace,avarikioti2023divide, byzcuit, byshard, stefoexecuting, cerberus, shaper, DangDLCLO19}, in which the blockchain state is split into shards, with a subset of the validators handling each shard in parallel. However, inter-validator sharding has limitations related to security and performance. Firstly, sharding requires a sampling process from the full validator set to subsets of validators per shard, such that each shard has more than 2/3 honest members. These systems are thus less robust to adversarial attacks. 
For example Omniledger~\cite{omniledger} assumes a 25\% Byzantine adversary in order to provide sufficient 34\% security in all the sub-sampled shards. 
In the same vein, the adversary's adaptivity should be limited to once an epoch, as otherwise the adversary could target all its power in a single shard and compromise it. 
Finally, sharding is also challenging from a performance perspective, as transactions that span multiple shards require expensive and slow Byzantine-resilient atomic commit protocols~\cite{byzcuit}.



\subsection{Intravalidator Sharding with \sysname}
Through \sysname, we instead propose \textit{intravalidator sharding}, as illustrated in \Cref{fig:components}. 
Each validator consists of multiple \emph{\seqworkers} that collect transaction data based on the commit sequence from the \primary, similar to transaction dissemination workers in lazy blockchains like Tusk~\cite{narwhal}, Bullshark~\cite{bullshark}, and Shoal~\cite{shoal}.
\sysname innovates by distributing transaction execution on several \emph{\execworkers}. Each \execworker stores a subset of the state, executes a subset of the transactions, and contributes its memory and storage to the system.

In \sysname, the \primary only manages metadata (agreement on a sequence of batch digests) allowing it to scale to large volumes of batches and transactions~\cite{narwhal}. Actual batch storage is distributed among a potentially large number of \seqworkers. The key insight is that transaction execution is also distributed among numerous \execworkers, enabling horizontal scaling. As workers are added, the capacity to store state and process transactions increases.

\begin{figure}[t]
    \centering
    \includegraphics[width=0.7\textwidth]{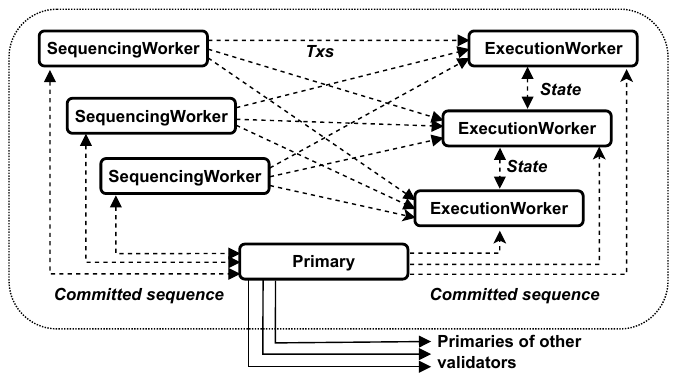}
    \caption{
        \sysname validator's components. Each validator is composed of several \seqworkers to fetch and persist the client's transaction, one \primary to run Byzantine agreement on metadata, and several \execworkers to execute transactions. Each component may run on dedicated machines or be collocated with other components. Dotted arrows indicate internal messages exchanged between the components of the validator (localhost or LAN) and solid arrows indicate messages exchanged with the outside world (WAN).
    }
    \label{fig:components}
\end{figure}

\para{Sharding strategy}
\sysname uses its \seqworkers and its \execworkers to operate two levels of sharding.
(i) \sysname shards transaction data among its \seqworkers. Transactions batches (and thus clients' transactions) are assigned to \seqworkers deterministically based on their digest. \seqworker can be seen as architecturally equivalent to the worker machines used by lazy blockchains to decouple dissemination (performed by workers) from ordering (performed by the \primary).
All transactions of a batch are persisted by the same \seqworker. Each \seqworker maintains a key-value store
$\batchdb[\batchid] \rightarrow \batch$ 
mapping the batch digests $\batchid$ to each batch handled by the \seqworker.
(ii) Additionally, \sysname shards its state among its \execworkers. Each \execworker is responsible for a disjoint subset of the objects in the system (composing the state); objects are assigned to \execworkers based on their collision-resistant identifier $\objectid$. Every object in the system is handled by exactly one (logical) \execworker.

\section{The \sysname System} \label{sec:design}

\Cref{fig:overview} shows the transaction life cycle in \sysname, from sequencing to execution.
The \primary sends the committed sequence to all  \seqworkers and \execworkers~(\one). Below, we outline the core \sysname protocol at steps \two, \three, \four, and \five of \Cref{fig:overview}. 
\ifpublish
Appendix~\ref{app:algorithms} presents the full algorithms.
\fi

\begin{figure*}[t]
    \centering
    \includegraphics[width=.8\textwidth]{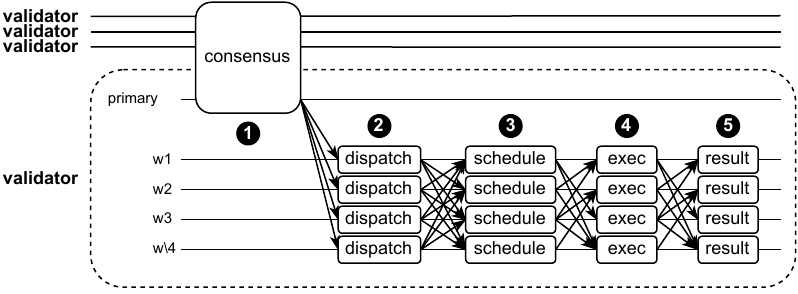}
    \caption{
        \sysname overview. Every validator runs with 5 machines: one machine running the \primary and 4 machines running workers. Each worker machine collocates 1 \seqworker and 1 \execworker. The \primary runs a Byzantine agreement protocol to sequence batch digests (\one). \seqworkers receive the committed sequence and load the data of the corresponding transactions from their storage (\two). Each \execworkers receiving these transactions assigns a lock to each object referenced by the transaction to schedule their execution (\three). A deterministically-selected \execworker eventually receives the object's data referenced by the execution and executes it (\four). Finally, the \execworker signals all \seqworkers to update their state with the results of the transaction's execution (\five).
    }
    \label{fig:overview}
    \vspace{1pt}
\end{figure*}

\execworkers maintain the following key-value stores:
\begin{itemize}
    \item $\objectsdb[\objectid] \rightarrow \object$ making all the objects handled by the \execworker accessible by their unique identifier.
    \item $\pendingdb[\objectid] \rightarrow [(op, [\tx])]$ mapping each object to a list of pending transactions $[\tx]$ referencing $\objectid$ in their read or write set and that are awaiting execution. The operation $op$ indicates whether the transaction may only Read ($\Read$) the object or whether it may also write ($\Write$) it. This map is used as a `locking' mechanism to track dependencies and determine which transactions can be executed in parallel. Entries relating to a transaction are removed from this map after its execution.
    \item $\missingdb[\objectid] \rightarrow [\tx]$ mapping objects that are missing from $\objectsdb$ to the transactions that reference them. It is used to track transactions that cannot (yet) be executed because they reference objects that are not yet available. It is cleaned after execution.
\end{itemize}

\para{Step~\two: Dispatch transactions}
At a high level, each \seqworker $i$ observes the commit sequence and loads from storage all the batches referenced by the committed sequence that they hold in their $\batchdb_i$ store (and ignores the others). The \seqworker then parses each transaction of the batch (in the order specified by the batch) to determine which objects it contains. At the end of this process, \seqworker $i$ composes one $\propose$ for each \execworker $j$ of the validator:
$\propose_{i,j} \gets (\batchid, \batchidx, \\T)$.
The message contains the batch digest $\batchid$, an index $\batchidx$ uniquely identifying the batch in the global committed sequence and a list of transactions $T$ referencing at least one object handled by worker $j$. If no transactions affect worker $j$, the worker still receives an empty message so it can proceed. 

\para{Step~\three: Schedule execution}
Each \execworker $j$ awaits one $\propose$ from each \seqworker. It then parses every transaction $\tx$ included (in order) and extracts objects in $\tx$'s read set $\R_j$ and write set $\W_j$ managed by \execworker $j$ (and ignores the other objects that it does not handle).

\begin{figure}[t]
    \centering
    \includegraphics[width=0.75\textwidth]{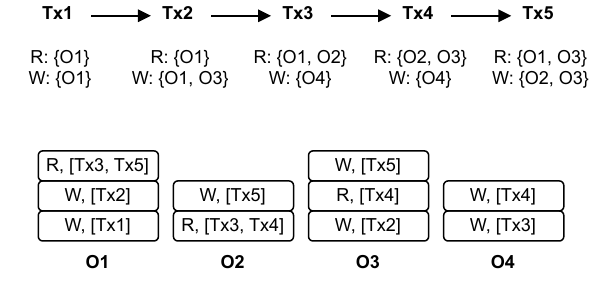}
    \caption{
        Example snapshot of the $\pendingdb$ queues of an \execworker. \sysname schedules the execution of the sequence $[\tx_1, \tx_2, \tx_3, \tx_4, \tx_5]$. The \execworker stores $\tx_1$ as $(\Write, [\tx_1])$ in the queue of $\objectid_1$ as it only mutates $\objectid_1$. $\tx_2$ then mutates $\objectid_1$ and writes $\objectid_3$; it is thus store in the queue of $\objectid_1$ (implicitly taking $\tx_1$ as dependency) and $\objectid_3$. $\tx_3$ schedules a read for both $\objectid_1$ and $\objectid_2$ and a write for $\objectid_4$. $\tx_4$ reads $\objectid_2$ (it can thus read $\objectid_2$ in parallel with $\tx_3$, registering $(\Read, [\tx_3, \tx_4])$ in the queue of $\objectid_2$) and $\objectid_3$, and writes $\objectid_4$. Finally $\tx_5$ reads $\objectid_1$ (it can thus read $\objectid_1$ in parallel with $\tx_3$), writes $\objectid_2$ and mutates $\objectid_3$.
    }
    \label{fig:pending}
    \vspace{1pt}
\end{figure}

\Cref{fig:pending} illustrates an example snapshot of the $\pendingdb_j$ store of a validator.
\execworkers append every object of the write set $\W_j$
to their local $\pendingdb_j$ indicating that $\tx$ may mutate $\objectid$:
$\pendingdb_j[\objectid] \gets \pendingdb_j[\objectid] \cup (\Write, \tx)$.
The position of $\tx$ in the $\pendingdb_j$ indicates that $\tx$ can only write $\objectid$ after all transactions appended before in $\pendingdb_j[\objectid]$ are executed, essentially indicating a write-after-write (or write-after-read) dependency.

\execworkers additionally register reads performed by $\tx$ on an object id by looking at the latest entry in $\pendingdb_j[\objectid]$.
If the entry is a write then they append a new entry:
$\pendingdb_j[\objectid] \gets \pendingdb_j[\objectid] \cup (\Read, \tx)$,  indicating a read-after-write dependency.
However, if the entry is a read then the transaction $\tx$ may be executed in parallel with any other transaction $\tx'$ also reading $\objectid$.
\execworkers thus modify the latest entry of the storage to reflect this possibility by setting $\tx$ and $\tx'$ at the same height in the $\pendingdb_j$ store:
$\pendingdb_j[\objectid][-1] \gets (\Read, [\tx', \tx])$.

A transaction $\tx$ is ready to be executed when it reaches the head of the pending lists of all the objects it references.
At this point, the \execworker loads from its $\objectsdb_j$ store all the objects data it handles:
$O_j \gets \{ \objectsdb[\objectid] \text{ s.t. }\\ \objectid \in \Call{HandledObjects}{\tx} \}$.
It then composes a $\ready$ for the dedicated \execworker that was selected to execute $\tx$: 
$\ready_{j} \gets (\tx, O_j)$. 
The message contains the transaction $\tx$ to execute, and a list of object data ($O_j$) referenced by the part of the read and write set of $\tx$ handled by \execworker $j$.

If an object referenced by $\tx$ is absent from the \execworker's local $\objectsdb_j$ store, the \execworker waits until it all transactions sequenced before \tx are executed 
and then sends $\none$ instead of the object's data. This signals that $\tx$ is malformed and references non-existent objects or objects that should have been created but the origin transaction failed.

\para{Step~\four: Execute transactions}
Upon receiving a $\ready$ message, an \execworker waits for one $\ready$ from all other \execworkers handling at least one object referenced by $\tx$.
At this point, the set of $\ready$ provides the \execworker with the objects' data behind all objects referenced by $\tx$ (or $\perp$ if missing). If all object data are available,
\tx is executed; otherwise, it is aborted.
Executing a transaction produces a set of objects to mutate or create $O$ and a set of object ids to delete $I$:
$(O, I) \gets \exec{\tx, O'}$.
The \execworker then prepares a $\result$ for all \execworkers. For \execworkers whose objects are not affected by $\tx$ this serves as a heartbeat message whereas for those whose objects are mutated, created or deleted by the transaction execution it informs them to update their object store $\objectsdb$ accordingly.
If \tx aborts, the worker sends a $\result$ with empty $O, I$.

\para{Step~\five: Handle results}
When an \execworker receives a $\result$, it:
(i) persists locally the fact that the transaction has been executed by advancing a watermark keeping track of all executed transactions;
(ii) updates each object into its local $\objectsdb$ store including deletions;
and (iii) removes all occurrences of the transaction from its $\pendingdb$ store.
It then tries to trigger the execution of the next transactions in the queues.
%

\section{Crash Fault Tolerance} \label{sec:cft-short}

\Cref{sec:design} presents the design of \sysname assuming all data structures are in-memory. However, critical validator components inevitably fail over time. To handle this, \sysname adopts a simple replication architecture, dedicating multiple machines to each \execworker. This internal replication allows the validator to continue operating despite crash faults. \sysname does not replicate the \primary, which handles only lightweight operations (and holds the signing key), nor does it replicate \seqworkers, which perform stateless work and can be rebooted from the latest persisted sequence number. We briefly detail our replication protocol here and defer more details to 
\ifpublish 
Appendices~\ref{sec:cft} and \ref{sec:detailed-recovery}.
\fi

\subsection{Internal Replication}

\Cref{fig:replication-main} illustrates the replication strategy of \sysname. Each \execworker is replaced by $n_e=2f_e+1$ \execworkers. \sysname tolerates up to $f_e$ simultaneous crash faults in a set of $n_e$ replicated \execworkers. These replicas form a grid: each column represents replicas of a single shard; each row is a cluster containing exactly one replica from each shard. Within each cluster, workers exchange reads and maintain a consistent view of the object store.

The naïve way to achieve such reliability would be to run a black-box replication engine like Paxos~\cite{lamport2001paxos} which is also the proposal of the state-of-the-art~\cite{calvin}. \sysname however greatly simplifies this process by leveraging (i) the \primary as a coordinator between the workers' replicas, (ii) external validators holding the blockchains state and the commit sequence, and (iii) the fact that execution is deterministic (given the commit sequence).

\begin{figure}[t]
    \centering
    \includegraphics[width=0.5\columnwidth]{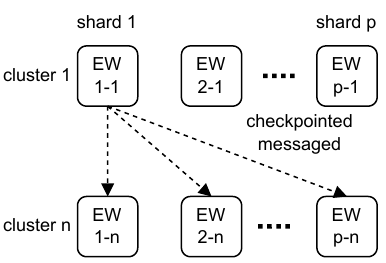}
    \caption{
        Replication scheme for \execworkers. The object store is partitioned into shards, and each shard is replicated $n_e$-fold. Each row represents a cluster, and \execworkers within a cluster coordinate to process transactions. During normal operation, the only communication between clusters is the sending of checkpoints.
    }
    \label{fig:replication-main}
\end{figure}

\subsection{Normal Operation}

Within each cluster, replicated \execworkers run the same core protocol as the unreplicated case. Inter-cluster communication is minimal, except for checkpoint updates.
To enable recovery, each worker keeps: (1) a buffer of outgoing \ready instances (the reads it has served), which can be replayed if messages are lost, and (2) a set of checkpoints, each representing a consistent, on-disk snapshot of the local object store. Checkpoints are the only persistent state.

\para{Garbage collection}
Once a checkpoint is deemed stable---i.e., a quorum of $f_e + 1$ replicas in every shard confirm they have persisted a checkpoint after a certain transaction index---old checkpoints and buffered messages prior to that index are garbage-collected.

\para{Bounding memory use}
Even with garbage collection, differences in execution speeds of workers can prevent checkpoints from being safely garbage-collected, and thus lead to unbounded memory use. To prevent this, \sysname enforces a maximum of $c$ checkpoints per worker. If a worker reaches this limit (e.g., $c=2$), it must pause execution until it can safely discard an older checkpoint. This prevents unlimited checkpoint buildup and ensures that clusters can proceed without one shard outpacing the others indefinitely. Typically, $c=2$ strikes a balance between performance and resource usage, letting faster clusters keep going as long as they are within one checkpoint boundary of slower ones.

\subsection{Failure Recovery}
\sysname uses two mechanisms to recover from failures:
(1) \textit{reconfiguration}, a rapid process that does not reduce the system’s throughput, but requires roughly synchronized clusters; and (2) \textit{checkpoint synchronization}, a slower procedure that coordinates multiple clusters if reconfiguration fails.
If both approaches fail, the system can still recover from other blockchain validators.

\para{Recovery through reconfiguration}
When an \execworker crashes, workers which rely on it for reads may be unable to proceed. They detect the crash and establish a new connection with a replacement replica. Typically, this requires just two round trips: one to identify a new node that can provide reads, and another to complete the handshake. Other clusters keep executing, so throughput remains unaffected as long as the failure threshold $f_e$ is not exceeded. 
\ifpublish
\Cref{sec:detailed-recovery} 
\fi
contains the full algorithm.

\para{Recovery through checkpoint synchronization}
If a worker is too far behind for reconfiguration alone, it triggers a synchronization process for itself and its peers, which fetch the latest checkpoint from an up-to-date replica.
Once all peers reach the same state, the worker re-establishes missing members in its cluster (via reconfiguration).
This cascades recovery across clusters that depend on the slow worker, ensuring no cluster loses liveness if another shard ``fast-forwards'' its state.

\para{Disaster recovery}
If an entire cluster is lost beyond the threat model of~\Cref{sec:model}, the system can recover by booting a new cluster with the same peers set. This new cluster retrieves the system state from other validators, which store stable checkpoints. Though it requires wide-area network communications and is slower, this worst-case path ensures \sysname remains operable even with minimal replication (e.g., $f_e=1$).

\section{Dynamic Reads and Writes} \label{sec:dynamic}
In most deterministic execution engines~\cite{solana-vm, fuelvm, calvin, faleiro2015bohm}, transactions must specify the exact data they read and write.
This constraint limits developers and encourages the over-prediction of read/write sets to ensure successful execution.
In distributed execution, the problem is exacerbated by the need to transmit the data between \execworkers. This means that we might need to
transmit large read/write sets between computers in order to access a single item (e.g., transfer a full array to dynamically access one cell).

\sysname supports dynamic reads/writes but confines them to parent-child object hierarchies.
A child object is an object that is owned by another object, the parent.
An example parent-child relationship is that between a dynamically allocated array and its individual cells.

In \sysname, a child object can only be accessed if the root object (the top-level object in a hierarchy of potentially numerous parents) is included in the transaction and the transaction has permission to access the root.
This setup avoids overpredictions by allowing transactions to handle unexpected data accesses with minimal algorithmic changes.

One of the required modifications is to retain the reads in the queues until the transaction execution is completed.
However, this leads to a loss of parallelism since we are unable to write a new version of an object until all transactions reading the previous version have finished. We resolve this false sharing situation without bloating memory usage in two ways.
First, we treat every version of an object as a new object; this means that the queues in \Cref{fig:pending} are per $(\objectid,\objectversion)$ instead of per $\objectid$.
Therefore, each queue consists of a single write as the initial transaction, followed by potentially several reads.
This resolves the false sharing as future versions of an object initialize new queues and can proceed independently of whether the previous version is still locked because of a dynamic read operation.
Unfortunately, this leads to objects potentially being written out of order, which could pollute our state and make consistent recovery from crashes impossible. For this reason, our second modification is buffering writes so that they are written to disk in order by leveraging the crash-recover algorithm in \Cref{sec:cft-short}.
\ifpublish
Appendix~\ref{sec:child-details} 
\fi
provides further details on how we handle child objects, complete algorithms, and formal proofs.


\para{Algorithm modifications}
\sysname handles the state of child objects like any other object: they are assigned to \execworkers that maintain their pending queues.
The \execworkers schedule the execution of root objects as usual after processing a $\propose$ 
\ifpublish
(\Cref{alg:process-propose})
\fi 
by updating the queues of all the objects that the transaction directly references. This means that they update the queues of (potentially) root objects as well as the queues of (potentially currently undefined) child objects. The security of this process is ensured by following the same procedure as for object creation. Hence, the \execworker will either create these objects or garbage-collect them.
Finally, when the transaction is ready for execution, either a previous transaction would have transferred ownership of child objects to the parent or the transaction would abort at execution.

On receiving a $\ready$, the \execworker starts execution. If it detects a new child object, it pauses and sends a $\updatepropose$ carrying an \emph{augmented transaction} $\txaugmented$, which includes the child objects ID in its read/write sets. This $\updatepropose$ message is sent to shards handling one of the (newly discovered) child objects. This is safe because the parent is already locked, implicitly locking the child. If the transaction is not done, subsequent parent writes go to distinct queues, enabling on-demand multi-version concurrency.

Upon receiving $\updatepropose$ with $\txaugmented$ , the \execworker replaces $\tx$ in its queues with $\txaugmented$ and adds $\txaugmented$ to the queues of any newly discovered child objects. When $\txaugmented$ reaches the front of every involved queue, it re-attempts execution. Eventually, the protocol identifies every child object that the transaction dynamically accesses, and $\txaugmented$ contains their explicit ids. At this point, the transaction execution can terminate successfully. 

\section{Implementation} \label{sec:implementation}
We implement a networked multi-core \sysname execution engine in Rust on top of the Sui blockchain~\cite{sui}. As a result, our implementation supports Sui-Move~\cite{sui-move}. We made this choice because Sui-Move is a simple and expressive language that is easy to reason about, provides a well-documented transaction format explicitly exposing the input read and write set, and supports dynamic reads and writes.
Our implementation uses \texttt{tokio}~\cite{tokio} for asynchronous networking across the \sysname workers, utilizing low-level TCP sockets for communication without relying on any RPC frameworks.
While all network communications in our implementation are asynchronous, the core logic of the execution worker runs synchronously in a dedicated thread. This approach facilitates rigorous testing, mitigates race conditions, and allows for targeted profiling of this critical code path.
In addition to regular unit tests, we created a command-line utility (called \emph{orchestrator}) designed to deploy real-world clusters of \sysname with workers distributed across multiple machines. The orchestrator has been instrumental in pinpointing and addressing efficiency bottlenecks.
We will open-source our \sysname implementation along with its orchestration utilities.%
\footnote{\url{https://github.com/mystenlabs/sui/tree/sharded-execution}}.
\section{Evaluation} \label{sec:evaluation}
We evaluate the performance of \sysname through experiments on Amazon Web Services (AWS) to show that
given a sufficiently parallelizable compute-bound load, the throughput of \sysname linearly increases with the number of \execworkers without visibly impacting latency. In order to investigate the spectrum of \sysname, we (a) run with transactions of increasing computational load and (b) create a contented workload that is not ideal for \sysname as it (i) increases the amount of communication among \execworkers and (ii) might increase the queuing delays in order to unblock later transactions.
We show the performance improvements of \sysname over the baseline execution engine of Sui~\cite{sui}.

\subsection{Experimental Setup}
We deploy \sysname on AWS, using \texttt{m5d.8xlarge} within a single datacenter (us-west-1). Each machine provides 10 Gbps of bandwidth, 32 virtual CPUs (16 physical cores) on a 2.5GHz, Intel Xeon Platinum 8175, 128GB memory, and runs Linux Ubuntu server 22.04. We select these machines because they provide decent performance, and are in the price range of `commodity servers'.

In all graphs, each data point represents median latency/throughput over a 5-minute run.
We instantiate one benchmark client collocated with each \seqworker submitting transactions at a fixed rate for a duration of 5 minutes. We experimentally increase the load of transactions sent to the systems, and record the throughput and latency of executed transactions. As a result, all plots illustrate the `steady state' latency of all systems under low load, as well as the maximal throughput they can serve, after which latency grows quickly. We vary the types of transactions throughout the benchmark to experiment with different contention patterns.

When referring to \emph{latency}, we mean the time elapsed from when the client submits the transaction until the transaction is executed. By \emph{throughput}, we mean the number of executed transactions over the entire duration of the run.

\subsection{Simple Transfer Workload}\label{sec:transfers}

\begin{figure}[tb]
    \centering
    \begin{subfigure}{0.49\textwidth}
        \includegraphics[width=\columnwidth]{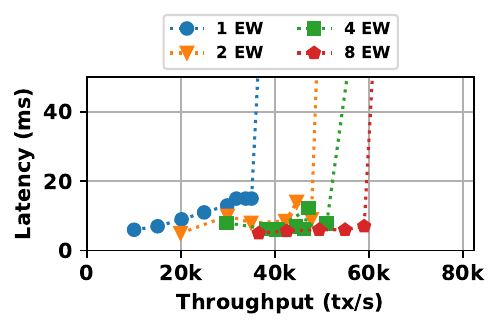}
        \caption{}
        \label{fig:lat-transfers}
    \end{subfigure}
    \hfill
    \begin{subfigure}{0.49\textwidth}
        \includegraphics[width=\columnwidth]{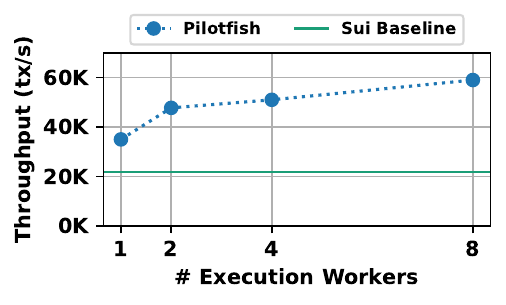}
        \caption{}
        \label{fig:scale-transfers}
    \end{subfigure}
    \caption{\sysname latency vs throughput (a) and scalability (b) with simple transfers.}
\end{figure}



\begin{figure*}[t]
    \centering

    \begin{subfigure}{0.32\textwidth}
        \includegraphics[width=\linewidth]{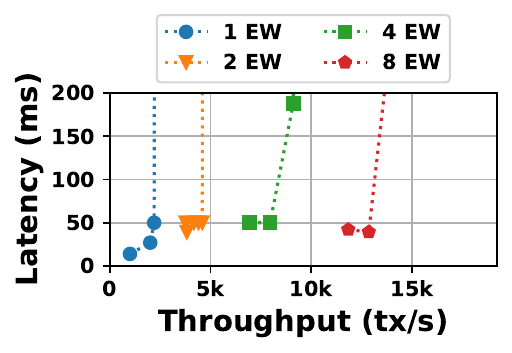}
        \caption{\texttt{Fib-2500}}
    \end{subfigure}
    \hfill
    \begin{subfigure}{0.32\textwidth}
        \includegraphics[width=\linewidth]{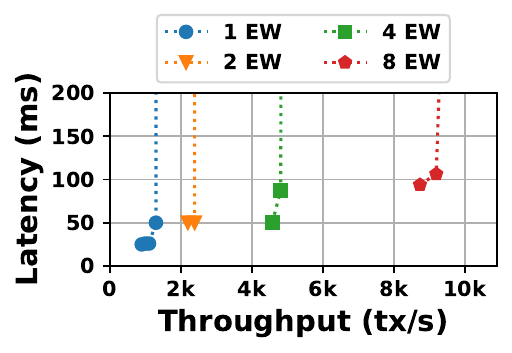}
        \caption{\texttt{Fib-5000}}
    \end{subfigure}
    \hfill
    \begin{subfigure}{0.32\textwidth}
        \includegraphics[width=\linewidth]{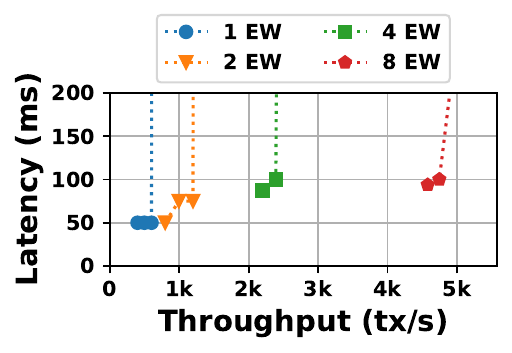}
        \caption{\texttt{Fib-10000}}
    \end{subfigure}
    \caption{\sysname latency vs. throughput for the heavy computation workloads.}
    \label{fig:lat-heavy}
\end{figure*}

\begin{figure}
    \centering
    \begin{subfigure}{0.49\textwidth}
        \includegraphics[width=\columnwidth]{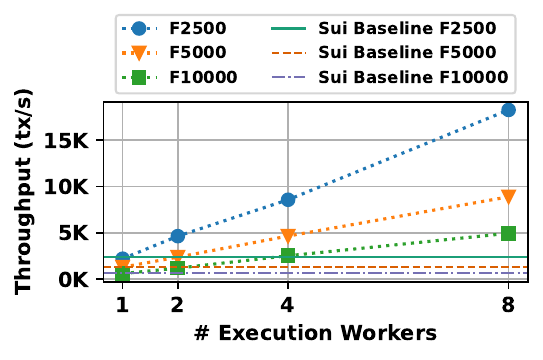}
        \caption{}
        \label{fig:scale-heavy}
    \end{subfigure}
    \hfill
    \begin{subfigure}{0.49\textwidth}
        \includegraphics[width=\columnwidth]{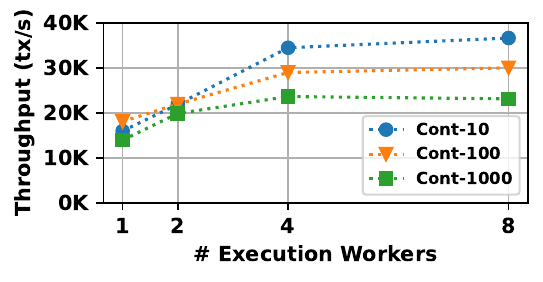}
        \caption{}
        \label{fig:scale-contended}
    \end{subfigure}
    \caption{(a) \sysname scalability with computationally heavy transactions. \texttt{F\{X\}} means that each transaction computes the $X$-th Fibonacci number. The horizontal lines show the single-machine throughput of the baseline on the same workloads. (b) \sysname scalability with condended transaction. Each transaction increments a counter. \texttt{Cont-X} means that for each counter we submit $X$ increment transactions.}
\end{figure}



\begin{figure*}[t]
    \centering
    \begin{subfigure}{0.32\textwidth}
        \includegraphics[width=\linewidth]{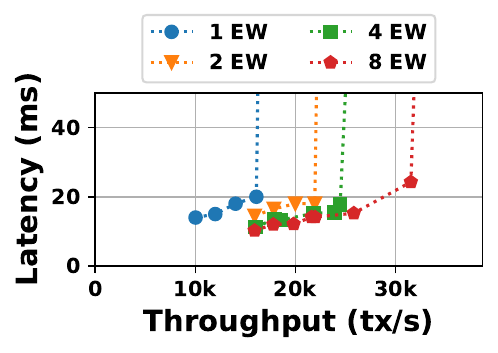}
        \caption{10 transactions/counter}
    \end{subfigure}
    \hfill
    \begin{subfigure}{0.32\textwidth}
        \includegraphics[width=\linewidth]{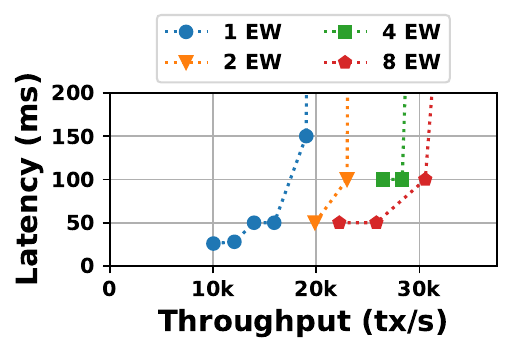}
        \caption{100 transactions/counter}
    \end{subfigure}
    \hfill
    \begin{subfigure}{0.32\textwidth}
        \includegraphics[width=\linewidth]{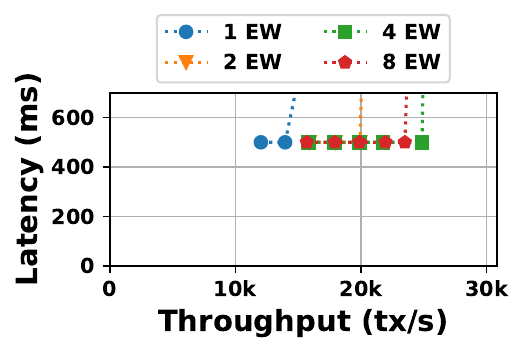}
        \caption{1000 transactions/counter}
    \end{subfigure}
    \caption{\sysname latency vs throughput for the contended workloads. Please note the different $y$ axis ranges between the three cases.}
    \label{fig:lat-contention}
    \vspace{1pt}
\end{figure*}

In this workload, each transaction is a simple transfer of coins between objects. No two transactions conflict; each transaction operates on a different set of objects from the other transactions. Thus, this workload is completely parallelizable. \Cref{fig:lat-transfers} shows latency vs throughput of \sysname on this workload with $1$, $2$, $4$ and $8$ \execworkers, and \Cref{fig:scale-transfers} shows how \sysname's maximum throughput scales when varying the number of \execworkers. 

\Cref{fig:scale-transfers} includes as baseline the throughput of the Sui execution engine.\footnote{We obtain the baseline by running Sui's single node benchmark with the \texttt{with-tx-manager} option.} 
Since the Sui transaction manager currently relies on stable storage, whereas \sysname is in-memory, this baseline is a lower bound on the expected performance of our system, when using a single \execworker.

We observe that in all cases, \sysname maintains a $20$ms latency envelope for this workload. Note that latency exhibits a linear increase as the workload grows for a single \execworker, primarily because of transaction queuing. More specifically, we see that a single machine does not have enough cores to fully exploit the parallelism of the workload, so some transactions must wait to get scheduled. This effect no longer exists for higher numbers of \execworkers, showing that more hardware has a beneficial effect on service time.

\sysname scales up to around $60k$ tx/s. In contrast, the Sui baseline can only process around $20k$ tx/s as it cannot leverage the additional hardware. \sysname thus exhibits a $3\times$ throughput improvement over the baseline.

\sysname's scalability is not perfectly linear in this workload; in particular, it becomes less steep after $2$ \execworkers. This is because the simple transfers workload is computationally light, and so the system is not compute-bound. Thus adding more resources no longer improves performance proportionally. \Cref{sec:heavy-workload} illustrates the advantages of increasing the number of \execworkers further when the workload is compute-bound.


\subsection{Computationally-Heavy Workload} \label{sec:heavy-workload}

We study the scenario when the workload remains compute-bound even at higher numbers of \execworkers. In this workload, transactions are computationally heavy. To achieve this, each transaction merges two coins and then iteratively computes the $X$th Fibonacci number, where $X$ is a configurable parameter. We study the behavior of \sysname for $X \in \{2500, 5000, 10000\}$. This workload is also perfectly parallel: transactions operate on disjoint sets of coins and thus do not conflict. \Cref{fig:lat-heavy} and \Cref{fig:scale-heavy} show the results: latency vs throughput and throughput scalability of \sysname, respectively. \Cref{fig:scale-heavy} includes the behavior of Sui on the same workloads, as a baseline.

As expected the performance of \sysname is on par with the Sui baseline for all three computation intensities when running on a single \execworker. However, when computing resources are the bottleneck, \sysname scales linearly as more resources are added to the system. As a result, \sysname can process 20k, 10k, and 5k tx/s when setting $X=2500$, $X=5000$, and $X=10000$, respectively, while maintaining the latency at around 50 ms.
In contrast, the throughput of the baseline execution engine of Sui remains set to a maximum of 2,5k, 1k, and 500 tx/s (with respectively $X=2500$, $X=5000$, and $X=10000$) as it is unable to take advantage of the additional hardware.
As a result, \sysname can process about 10x more transactions than the Sui baseline.


\subsection{Contended Workload}

We study the behavior of \sysname when the workload is no longer perfectly parallelizable. To achieve this, we introduce contention by making transactions operate on non-disjoint sets of objects. More concretely, in this workload each transaction increments a counter; for each counter, we generate a configurable number $Y$ of transactions that increment it. Thus, on average, each transaction needs to wait behind $Y/2$ other transactions in its counter's queue, before being able to execute. In our experiments, $Y \in \{10, 100, 1000\}$. The results are shown in \Cref{fig:lat-contention} and \Cref{fig:scale-contended}.
\sysname reaches a throughput of 35k, 30k, and 22k tx/s for $Y=10$, $Y=100$, and $Y=1000$ when operating with 4 \execworkers.
For this workload, for technical reasons,\footnote{In Sui, each transaction expects object references for all input objects. Each object reference is computed based on the last transaction to modify the object. Therefore, it is difficult to pre-generate more than one valid transaction for the same object, before the experiment starts, because correct object references cannot be predicted.} we could not include a Sui baseline.

As expected, we observe that as we increase contention, latency increases due to  queueing (up to $500$ms for $Y=1000$) and throughput decreases. Nonetheless, \sysname is able to scale to $4$ \execworkers. Similarly to the simple transfer workload (\Cref{sec:transfers}), this workload is not compute-bound, so adding compute beyond $4$ \execworkers no longer improves performance proportionally.

\section{Related Work} \label{sec:related}

\setlength{\tabcolsep}{0.5em}
\begin{table}[t]
\caption{Comparison against existing deterministic approaches}
\centering
\begin{tabular}{ccccc}
\toprule
 & \textbf{Distributed} & \textbf{\makecell{Crash \\ Tolerance}} & \textbf{\makecell{Dynamic \\ RW Set}} & \textbf{\makecell{No CC \\ Aborts}} \\
\midrule\midrule
\textsc{Bohm}~\cite{faleiro2015bohm} & \xmark & \xmark & \xmark & \cmark \\
PWV~\cite{faleiro2017pwv} & \xmark & \cmark & \xmark & \cmark\\
QueCC~\cite{quecc}  & \xmark  & \xmark  & \xmark & \cmark  \\
SLOG~\cite{slog}   & \cmark   & \cmark   &  \xmark  & \xmark  \\
Q-Store~\cite{qstore} & \cmark & \cmark & \xmark & \cmark \\
Calvin~\cite{thomson2014calvin} & \cmark & \cmark & \xmark & \cmark\\
Aria~\cite{lu2020aria}   & \cmark & \cmark & \cmark & \xmark   \\
Lotus~\cite{lotus} & \cmark & \cmark & \cmark & \xmark \\
\midrule
\sysname (this work) & \cmark & \cmark & \cmark & \cmark \\
\bottomrule
\end{tabular}
\label{tab:related}
\end{table}

\para{Parallel blockchain executors}
The main proposals in this area are those of Solana~\cite{solana-parallel}, Aptos~\cite{gelashvili2023blockstm}, and Sui~\cite{sui-lutris}. Solana~\cite{solana-parallel} requires every transaction to fully specify its read and write sets, so it cannot support dynamic accesses in the same way as \sysname. Aptos uses the Block-STM~\cite{gelashvili2023blockstm} system for parallel transaction execution. Block-STM is designed with a focus on single-machine, multi-threaded performance, and it is unclear how or if its design can be extended to the scale-out, distributed deployment that \sysname targets. For instance, Block-STM executes transactions speculatively, and retries transactions which fail validation. This approach works well in a shared memory environment, where retries are relatively inexpensive, but it is not clear if it can be applied in a distributed environment, where retries are much more costly due to higher communication latency. Futhermore, BlockSTM focuses on a per-block execution model which requires large blocks to optimize throughput, at the expense of latency. By contrast, \sysname uses a streaming execution model that allows for low latency regardless of throughput. Finally, Sui~\cite{sui-lutris} implicitly handles synchronization and scheduling through the tokio runtime~\cite{tokio}: a tokio task is spawned for each Sui transaction; this task waits for the transaction's dependencies to be satisfied (i.e. the required object versions to be available), and then executes the transaction in parallel with other tasks. It is unclear how to directly extend this approach to multiple machines, as required by \sysname.

\para{Deterministic databases}
\sysname is similar to deterministic database systems~\cite{thomson2010deterministic} that employ an order-then-execute approach. 
\Cref{tab:related} summarizes the main differences between \sysname and existing deterministic approaches. As \Cref{tab:related} shows, \sysname is the first distributed, crash fault tolerant deterministic execution engine that tolerates partially unspecified read/write sets and eliminates concurrency-control-related aborts. The closest works to \sysname are Calvin~\cite{calvin}, Aria~\cite{lu2020aria} and Lotus~\cite{lotus}. Calvin~\cite{thomson2014calvin} proposes the use of consensus to address crashes, which in our setting is overkill since the blockchain already provides sufficient determinism to recover without strong coordination. Aria and Lotus differ from \sysname by not establishing a total order on transactions before execution, which can lead to some transactions aborting due to conflicts; such transactions have to be retried later, increasing latency.

\ifpublish
\para{zkVMs and layers 2}
Recently, layer-2 solutions~\cite{poon2016bitcoin} and especially the use of zkVMs~\cite{ben2018scalable} have been suggested as a way to speed up the execution of blockchains. ZkVMs reduce the problem of execution to that of one prover replica and many verifier replicas that do not re-execute but simply check the proofs. These solutions are not comparable to \sysname as they are actually different ways to do deterministic execution and, as a result, compatible with \sysname, which could in principle deploy a zkVM instead of the MoveVM. We opted to use these execution engines in our design because zkVMs do not support parallelism due to the way in which proofs are generated. If, in the future, parallel zkVMs become possible, then \sysname could be adapted to help them scale out.
\fi

\ifpublish
    \section*{Acknowledgments}
    This work is funded by Mysten Labs.
    We thank George Danezis and Dahlia Malkhi for their feedback on the early manuscript.
    We also extend our thanks Mark Logan, Xun Li, and Mingwei Tian, who provided insights on the bottlenecks of the Sui blockchain execution engine. Special thanks to Xun for the implementation of the Sui workload generator which was very useful for our experiments.
    Special thanks to Marios Kogias for suggesting the idea of a queue-based locking mechanism to parallelize execution within a machine.
\fi

\ifpublish
\bibliographystyle{splncs04}
\fi
\bibliography{main}

\begin{thebibliography}{10}
\providecommand{\url}[1]{\texttt{#1}}
\providecommand{\urlprefix}{URL }
\providecommand{\doi}[1]{https://doi.org/#1}

\bibitem{crash-recovery-fd}
Aguilera, M.K., Chen, W., Toueg, S.: Failure detection and consensus in the crash-recovery model. Distributed Comput.  \textbf{13}(2),  99--125 (2000)

\bibitem{sinfonia}
Aguilera, M.K., Merchant, A., Shah, M.A., Veitch, A.C., Karamanolis, C.T.: Sinfonia: {A} new paradigm for building scalable distributed systems. {ACM} Trans. Comput. Syst.  \textbf{27}(3),  5:1--5:48 (2009). \doi{10.1145/1629087.1629088}, \url{https://doi.org/10.1145/1629087.1629088}

\bibitem{lazyledger}
Al-Bassam, M.: Lazyledger: A distributed data availability ledger with client-side smart contracts. arXiv preprint arXiv:1905.09274  (2019)

\bibitem{chainspace}
Al-Bassam, M., Sonnino, A., Bano, S., Hrycyszyn, D., Danezis, G.: Chainspace: A sharded smart contracts platform. arXiv preprint arXiv:1708.03778  (2017)

\bibitem{elasticitySOTA}
Al{-}Dhuraibi, Y., Paraiso, F., Djarallah, N., Merle, P.: Elasticity in cloud computing: State of the art and research challenges. {IEEE} Trans. Serv. Comput.  \textbf{11}(2),  430--447 (2018)

\bibitem{shaper}
Amiri, M.J., Agrawal, D., Abbadi, A.E.: Sharper: Sharding permissioned blockchains over network clusters. In: {SIGMOD} '21: International Conference on Management of Data, Virtual Event, China, June 20-25, 2021. pp. 76--88. {ACM} (2021), \url{https://doi.org/10.1145/3448016.3452807}

\bibitem{aptos}
Aptos: Aptos. \url{https://aptoslabs.com} (2024)

\bibitem{aptosreqs}
{Aptos Node Requirements}. \url{https://aptos.dev/en/network/nodes/validator-node/node-requirements}, accessed: 2024-08-02

\bibitem{avarikioti2023divide}
Avarikioti, Z., Desjardins, A., Kokoris-Kogias, L., Wattenhofer, R.: Divide \& scale: Formalization and roadmap to robust sharding. In: International Colloquium on Structural Information and Communication Complexity. pp. 199--245. Springer (2023)

\bibitem{aws-autoscaling}
{Automatically manage Amazon ECS capacity with cluster auto scaling}. \url{https://docs.aws.amazon.com/AmazonECS/latest/developerguide/cluster-auto-scaling.html}, accessed: 2024-08-02

\bibitem{mysticeti}
Babel, K., Chursin, A., Danezis, G., Kokoris-Kogias, L., Sonnino, A.: Mysticeti: Low-latency dag consensus with fast commit path. arXiv preprint arXiv:2310.14821  (2023)

\bibitem{prism}
Bagaria, V., Kannan, S., Tse, D., Fanti, G., Viswanath, P.: Prism: Deconstructing the blockchain to approach physical limits. In: Proceedings of the 2019 ACM SIGSAC Conference on Computer and Communications Security. pp. 585--602 (2019)

\bibitem{ben2018scalable}
Ben-Sasson, E., Bentov, I., Horesh, Y., Riabzev, M.: Scalable, transparent, and post-quantum secure computational integrity. Cryptology ePrint Archive  (2018)

\bibitem{sui-lutris}
Blackshear, S., Chursin, A., Danezis, G., Kichidis, A., Kokoris-Kogias, L., Li, X., Logan, M., Menon, A., Nowacki, T., Sonnino, A., et~al.: Sui lutris: A blockchain combining broadcast and consensus. arXiv preprint arXiv:2310.18042  (2023)

\bibitem{scalability}
Bondi, A.B.: Characteristics of scalability and their impact on performance. In: Second International Workshop on Software and Performance, {WOSP} 2000, Ottawa, Canada, September 17-20, 2000. pp. 195--203. {ACM} (2000). \doi{10.1145/350391.350432}, \url{https://doi.org/10.1145/350391.350432}

\bibitem{celestia}
Celestia: {The first modular blockchain network}. \url{https://celestia.org} (2022)

\bibitem{chandra96failure}
Chandra, T.D., Hadzilacos, V., Toueg, S.: The weakest failure detector for solving consensus. J. ACM  \textbf{43}(4),  685–722 (jul 1996). \doi{10.1145/234533.234549}, \url{https://doi.org/10.1145/234533.234549}

\bibitem{chandra1996}
Chandra, T.D., Toueg, S.: Unreliable failure detectors for reliable distributed systems. J. ACM  \textbf{43}(2),  225–267 (mar 1996). \doi{10.1145/226643.226647}, \url{https://doi.org/10.1145/226643.226647}

\bibitem{coinmarketcap}
{CoinMarketCap}. \url{http://www.coinmarketcap.com}, accessed: 2024-08-02

\bibitem{pnuts}
Cooper, B.F., Ramakrishnan, R., Srivastava, U., Silberstein, A., Bohannon, P., Jacobsen, H., Puz, N., Weaver, D., Yerneni, R.: {PNUTS:} yahoo!'s hosted data serving platform. Proc. {VLDB} Endow.  \textbf{1}(2),  1277--1288 (2008). \doi{10.14778/1454159.1454167}, \url{http://www.vldb.org/pvldb/vol1/1454167.pdf}

\bibitem{narwhal}
Danezis, G., Kokoris-Kogias, L., Sonnino, A., Spiegelman, A.: Narwhal and tusk: a dag-based mempool and efficient bft consensus. In: Proceedings of the Seventeenth European Conference on Computer Systems. pp. 34--50 (2022)

\bibitem{DangDLCLO19}
Dang, H., Dinh, T.T.A., Loghin, D., Chang, E., Lin, Q., Ooi, B.C.: Towards scaling blockchain systems via sharding. In: Proceedings of the 2019 International Conference on Management of Data, {SIGMOD} Conference 2019, Amsterdam, The Netherlands, June 30 - July 5, 2019. pp. 123--140. {ACM} (2019), \url{https://doi.org/10.1145/3299869.3319889}

\bibitem{elastras}
Das, S., Agrawal, D., Abbadi, A.E.: {ElasTraS}: An elastic, scalable, and self-managing transactional database for the cloud. {ACM} Trans. Database Syst.  \textbf{38}(1), ~5 (2013). \doi{10.1145/2445583.2445588}, \url{https://doi.org/10.1145/2445583.2445588}

\bibitem{dynamo}
DeCandia, G., Hastorun, D., Jampani, M., Kakulapati, G., Lakshman, A., Pilchin, A., Sivasubramanian, S., Vosshall, P., Vogels, W.: Dynamo: {Amazon}'s highly available key-value store. In: Proceedings of the 21st {ACM} Symposium on Operating Systems Principles 2007, {SOSP} 2007, Stevenson, Washington, USA, October 14-17, 2007. pp. 205--220. {ACM} (2007). \doi{10.1145/1294261.1294281}, \url{https://doi.org/10.1145/1294261.1294281}

\bibitem{move-vm}
Docs: {Move VM}. \url{https://docs.dfinance.co/move\_vm} (2023)

\bibitem{dwork1988consensus}
Dwork, C., Lynch, N., Stockmeyer, L.: Consensus in the presence of partial synchrony. Journal of the ACM (JACM)  \textbf{35}(2),  288--323 (1988)

\bibitem{evm}
{Ethereum Foundation}: {Ethereum Virtual Machine (EVM)}. \url{https://ethereum.org/en/developers/docs/evm/} (2023)

\bibitem{faleiro2015bohm}
Faleiro, J.M., Abadi, D.J.: Rethinking serializable multiversion concurrency control. Proc. VLDB Endow.  \textbf{8}(11),  1190–1201 (jul 2015). \doi{10.14778/2809974.2809981}, \url{https://doi.org/10.14778/2809974.2809981}

\bibitem{faleiro2017pwv}
Faleiro, J.M., Abadi, D.J., Hellerstein, J.M.: High performance transactions via early write visibility. Proc. VLDB Endow.  \textbf{10}(5),  613–624 (jan 2017). \doi{10.14778/3055540.3055553}, \url{https://doi.org/10.14778/3055540.3055553}

\bibitem{fuelvm}
Fuel: {The World's Fastest Modular Execution Layer}. \url{https://www.fuel.network} (2024)

\bibitem{dumbo-ng}
Gao, Y., Lu, Y., Lu, Z., Tang, Q., Xu, J., Zhang, Z.: Dumbo-ng: Fast asynchronous bft consensus with throughput-oblivious latency. In: Proceedings of the 2022 ACM SIGSAC Conference on Computer and Communications Security (2022)

\bibitem{gelashvili2023blockstm}
Gelashvili, R., Spiegelman, A., Xiang, Z., Danezis, G., Li, Z., Malkhi, D., Xia, Y., Zhou, R.: Block-stm: Scaling blockchain execution by turning ordering curse to a performance blessing. In: PPoPP '23. p. 232–244. Association for Computing Machinery, New York, NY, USA (2023). \doi{10.1145/3572848.3577524}, \url{https://doi.org/10.1145/3572848.3577524}

\bibitem{bullshark}
Giridharan, N., Kokoris-Kogias, L., Sonnino, A., Spiegelman, A.: Bullshark: Dag bft protocols made practical. arXiv preprint arXiv:2201.05677  (2022)

\bibitem{cerberus}
Hellings, J., Hughes, D.P., Primero, J., Sadoghi, M.: {Cerberus: Minimalistic Multi-shard Byzantine-resilient Transaction Processing}. J. Syst. Res.  \textbf{3}(1) (2023), \url{https://doi.org/10.5070/sr33161314}

\bibitem{byshard}
Hellings, J., Sadoghi, M.: {ByShard}: sharding in a {Byzantine} environment. {VLDB} J.  \textbf{32}(6),  1343--1367 (2023), \url{https://doi.org/10.1007/s00778-023-00794-0}

\bibitem{elasticity}
Herbst, N.R., Kounev, S., Reussner, R.H.: Elasticity in cloud computing: What it is, and what it is not. In: Kephart, J.O., Pu, C., Zhu, X. (eds.) 10th International Conference on Autonomic Computing, ICAC'13, San Jose, CA, USA, June 26-28, 2013. pp. 23--27. {USENIX} Association (2013)

\bibitem{howard2016flexible}
Howard, H., Malkhi, D., Spiegelman, A.: Flexible paxos: Quorum intersection revisited (2016)

\bibitem{omniledger}
Kokoris-Kogias, E., Jovanovic, P., Gasser, L., Gailly, N., Syta, E., Ford, B.: Omniledger: A secure, scale-out, decentralized ledger via sharding. In: 2018 IEEE symposium on security and privacy (SP). pp. 583--598. IEEE (2018)

\bibitem{cassandra}
Lakshman, A., Malik, P.: Cassandra: structured storage system on a {P2P} network. In: Tirthapura, S., Alvisi, L. (eds.) Proceedings of the 28th Annual {ACM} Symposium on Principles of Distributed Computing, {PODC} 2009, Calgary, Alberta, Canada, August 10-12, 2009. p.~5. {ACM} (2009). \doi{10.1145/1582716.1582722}, \url{https://doi.org/10.1145/1582716.1582722}

\bibitem{lamport2001paxos}
Lamport, L.: Paxos made simple. ACM SIGACT News (Distributed Computing Column) 32, 4 (Whole Number 121, December 2001) pp. 51--58 (2001)

\bibitem{lu2020aria}
Lu, Y., Yu, X., Cao, L., Madden, S.: Aria: A fast and practical deterministic oltp database. Proc. VLDB Endow.  \textbf{13}(12),  2047–2060 (jul 2020). \doi{10.14778/3407790.3407808}, \url{https://doi.org/10.14778/3407790.3407808}

\bibitem{two-phase}
Mohan, C., Lindsay, B.G., Obermarck, R.: Transaction management in the r* distributed database management system. {ACM} Trans. Database Syst.  \textbf{11}(4),  378--396 (1986). \doi{10.1145/7239.7266}, \url{https://doi.org/10.1145/7239.7266}

\bibitem{sui}
{Mysten Labs}: {Build without boundaries}. {https://sui.io} (2022)

\bibitem{sui-move}
{Mysten Labs}: {Move Concepts}. \url{https://docs.sui.io/concepts/sui-move-concepts} (2023)

\bibitem{ebb-and-flow}
Neu, J., Tas, E.N., Tse, D.: Ebb-and-flow protocols: A resolution of the availability-finality dilemma. In: 2021 IEEE Symposium on Security and Privacy (SP). pp. 446--465. IEEE (2021)

\bibitem{snap-and-chat}
Neu, J., Tas, E.N., Tse, D.: The availability-accountability dilemma and its resolution via accountability gadgets. In: International Conference on Financial Cryptography and Data Security. pp. 541--559. Springer (2022)

\bibitem{ongaro2015raft}
Ongaro, D., Ousterhout, J.: The raft consensus algorithm. Lecture Notes CS  \textbf{190}, ~2022 (2015)

\bibitem{poon2016bitcoin}
Poon, J., Dryja, T.: The bitcoin lightning network: Scalable off-chain instant payments. \url{https://lightning.network/lightning-network-paper.pdf} (2016)

\bibitem{solana-vm}
Protocol, S.: {What Is SVM - The Solana Virtual Machine}. \url{https://squads.so/blog/solana-svm-sealevel-virtual-machine} (2024)

\bibitem{qstore}
Qadah, T., Gupta, S., Sadoghi, M.: Q-store: Distributed, multi-partition transactions via queue-oriented execution and communication. In: Proceedings of the 23rd International Conference on Extending Database Technology, {EDBT} 2020, Copenhagen, Denmark, March 30 - April 02, 2020. pp. 73--84. OpenProceedings.org (2020), \url{https://doi.org/10.5441/002/edbt.2020.08}

\bibitem{quecc}
Qadah, T.M., Sadoghi, M.: {QueCC}: {A} queue-oriented, control-free concurrency architecture. In: Proceedings of the 19th International Middleware Conference, Middleware 2018, Rennes, France, December 10-14, 2018. pp. 13--25. {ACM} (2018), \url{https://doi.org/10.1145/3274808.3274810}

\bibitem{slog}
Ren, K., Li, D., Abadi, D.J.: {SLOG:} serializable, low-latency, geo-replicated transactions. Proc. {VLDB} Endow.  \textbf{12}(11),  1747--1761 (2019). \doi{10.14778/3342263.3342647}, \url{http://www.vldb.org/pvldb/vol12/p1747-ren.pdf}

\bibitem{accordion}
Serafini, M., Mansour, E., Aboulnaga, A., Salem, K., Rafiq, T., Minhas, U.F.: Accordion: Elastic scalability for database systems supporting distributed transactions. Proc. {VLDB} Endow.  \textbf{7}(12),  1035--1046 (2014), \url{http://www.vldb.org/pvldb/vol7/p1035-serafini.pdf}

\bibitem{solana-parallel}
{Solana Foundation}: Sealevel---parallel processing thousands of smart contracts. \url{https://solana.com/news/sealevel---parallel-processing-thousands-of-smart-contracts} (2019)

\bibitem{solanareqs}
{Solana Validator Requirements}. \url{https://docs.solanalabs.com/operations/requirements}, accessed: 2024-08-02

\bibitem{byzcuit}
Sonnino, A., Bano, S., Al-Bassam, M., Danezis, G.: Replay attacks and defenses against cross-shard consensus in sharded distributed ledgers. In: 2020 IEEE European Symposium on Security and Privacy (EuroS\&P). pp. 294--308. IEEE (2020)

\bibitem{shoal}
Spiegelman, A., Aurn, B., Gelashvili, R., Li, Z.: Shoal: Improving dag-bft latency and robustness. arXiv preprint arXiv:2306.03058  (2023)

\bibitem{stefoexecuting}
Stefo, C., Xiang, Z., Kokoris-Kogias, L.: Executing and proving over dirty ledgers. In: International Conference on Financial Cryptography and Data Security. pp. 3--20. Springer (2023)

\bibitem{objectmodel}
{Sui}: {Object Model}. \url{https://docs.sui.io/concepts/object-model}, accessed: 2024-08-02

\bibitem{suireqs}
{Sui Validator Node Configuration}. \url{https://docs.sui.io/guides/operator/validator-config}, accessed: 2024-08-02

\bibitem{thomson2010deterministic}
Thomson, A., Abadi, D.J.: The case for determinism in database systems. Proc. VLDB Endow.  \textbf{3}(1–2),  70–80 (sep 2010). \doi{10.14778/1920841.1920855}, \url{https://doi.org/10.14778/1920841.1920855}

\bibitem{calvin}
Thomson, A., Diamond, T., Weng, S.C., Ren, K., Shao, P., Abadi, D.J.: Calvin: Fast distributed transactions for partitioned database systems. In: Proceedings of the 2012 ACM SIGMOD International Conference on Management of Data. p. 1–12. SIGMOD '12, Association for Computing Machinery, New York, NY, USA (2012). \doi{10.1145/2213836.2213838}, \url{https://doi.org/10.1145/2213836.2213838}

\bibitem{thomson2014calvin}
Thomson, A., Diamond, T., Weng, S.C., Ren, K., Shao, P., Abadi, D.J.: Fast distributed transactions and strongly consistent replication for oltp database systems. ACM Trans. Database Syst.  \textbf{39}(2) (may 2014). \doi{10.1145/2556685}, \url{https://doi.org/10.1145/2556685}

\bibitem{tokio}
{Tokio is an asynchronous runtime for the Rust programming language.} \url{https://tokio.rs/}, accessed: 2024-10-04

\bibitem{lotus}
Zhou, X., Yu, X., Graefe, G., Stonebraker, M.: Lotus: Scalable multi-partition transactions on single-threaded partitioned databases. Proc. {VLDB} Endow.  \textbf{15}(11),  2939--2952 (2022), \url{https://www.vldb.org/pvldb/vol15/p2939-zhou.pdf}

\end{thebibliography}

\ifpublish
\appendix
\bigskip
\section*{\hfill APPENDIX\hfill}
\section{Algorithms} \label{app:algorithms}
This section complements \Cref{sec:design} by providing detailed algorithms for the core components of \sysname.

\subsection{Detailed Algorithms}
The function $\Call{Handle}{\objectid}$ of \Cref{alg:process-batch} returns the \execworker that handles the specified object identifier $\objectid$. The function $\Call{Index}{\tx}$ in \Cref{alg:core} returns the index of the transaction \tx in the global committed sequence. The function $\Call{ID}{o}$ in \Cref{alg:process-result} returns the object id $\objectid$ of the object $o$.

\begin{algorithm}[h!]
    \caption{Process committed sequence (step \two of \Cref{fig:overview})}
    \label{alg:process-batch}
    \algfontsize
    \begin{algorithmic}[1]
        \Statex // Called by \seqworkers upon receiving the committed sequence.
        \Procedure{ProcessSequencedBatch}{$\batchid, \batchidx$}
        \State // Ignore batches for other workers.
        \If{$\Call{Handler}{\batchid} \neq \self$} \Return \EndIf
        \State
        \State // Make one propose message for each \execworker.
        \State $\batch \gets \batchdb[\batchid]$
        \For{$w \in  \text{\execworkers}$}
        \State $T \gets [ \tx \in \batch \text{ s.t. } \exists \objectid \in \tx \text{ s.t. } \Call{Handler}{\objectid}=w ]$ \label{alg:line:propose-creation-1}
        \State $\propose \gets (\batchidx, \batchid, T)$ \label{alg:line:propose-creation-2}
        \State $\Call{Send}{w, \propose}$
        \EndFor
        \EndProcedure
    \end{algorithmic}
\end{algorithm}

\begin{algorithm}[t]
    \caption{Process $\propose$ (step \three of \Cref{fig:overview})}
    \label{alg:process-propose}
    \algfontsize
    \begin{algorithmic}[1]
        \State $\loadedIdx \gets 0$ \Comment{All batch indices below this watermark are received}
        \State $\loaded \gets [\;]$ \Comment{Received batch indices}
        \Statex
        \Statex // Called by \execworkers upon receiving a $\propose$.
        \Procedure{ProcessPropose}{$\propose$}
        \State // Ensure we received one message per \seqworker
        \State $(\batchidx, \batchid, T) \gets \propose$
        \State $\loaded[\batchidx] \gets \loaded[\batchidx] \cup (\batchid, T)$
        \While{$\len{\loaded[\loadedIdx]} = |\text{\seqworkers}|$} \label{alg:line:wait-for-batches}
        \State $(\_, T) \gets \loaded[\loadedIdx]$
        \State $\loadedIdx \gets \loadedIdx+1$ \label{alg:line:update-loaded-idx}
        \State
        \State // Add the objects to their pending queues
        \For{$\tx \in T$} \label{alg:line:process-proposed-txs}
        \For{$\objectid \in \Call{HandledObjects}{\tx}$} \Comment{Defined in \Cref{alg:core}}
        \If{$\objectid \in \writeset{\tx}$}
        \State $\pendingdb[\objectid] \gets \pendingdb[\objectid] \cup (W, [\tx])$ \label{alg:line:schedule-write}
        \Else \Comment{$\objectid \in \readset{\tx}$}
        \State $(op, T') \gets \pendingdb[\objectid][-1]$
        \If{$op = W$}
        $\pendingdb[\objectid] \gets \pendingdb[\objectid] \cup (R, [\tx])$  \label{alg:line:schedule-read}
        \Else
        $\; \pendingdb[\objectid][-1] \gets (R, T' \cup \tx)$ \label{alg:line:schedule-parallel-read}
        \EndIf
        \EndIf
        \State
        \State // Try to execute the transaction
        \State $\Call{TryTriggerExecution}{\tx}$ \Comment{Defined in \Cref{alg:core}} \label{alg:line:trigger-execution-after-scheduling}
        \EndFor
        \EndFor
        \EndWhile
        \EndProcedure
    \end{algorithmic}
\end{algorithm}

\begin{algorithm}[t]
    \caption{Core functions}
    \label{alg:core}
    \algfontsize
    \begin{algorithmic}[1]
        \State $\executedIdx \gets 0$ \Comment{All $\tx$ indices below this watermark are executed} \label{alg:line:executed-idx}
        \State $\executed \gets \emptyset$ \Comment{Executed transaction indices}

        \Statex
        \Function{TryTriggerExecution}{$\tx$} \label{alg:line:try-trigger-execution}
        \State // Check if all dependencies are already executed
        \If{$\Call{HasDependencies}{\tx}$} \Return \EndIf \label{alg:line:check-dependencies}
        \State
        \State // Check if all objects are present
        \State $M \gets \Call{MissingObjects}{\tx}$ \label{alg:line:check-missing-objects}
        \If{$M \neq \emptyset$}
        \For{$\objectid \in M$} $\missingdb[\objectid] \gets \missingdb[\objectid] \cup \tx$ \EndFor \label{alg:line:track-missing}
        \State \Return
        \EndIf
        \State
        \State // Send object data to a deterministically-selected \execworker
        \State $worker \gets \Call{Handler}{\tx}$ \Comment{Worker handling the most objects of $\tx$}
        \State $O \gets \{ \objectsdb[\objectid] \text{ s.t. } \objectid \in \Call{HandledObjects}{\tx} \}$ \Comment{May contain $\perp$} \label{alg:line:load-objects}
        \State $\ready \gets(\tx, O)$
        \State $\Call{Send}{worker, \ready}$ \label{alg:line:send-ready}
        \State
        \State // Remove read-locks from the pending queues
        \For{$\objectid \in \readset{\tx}$}
        \State $T' \gets \Call{AdvanceLock}{\tx, \objectid}$ \label{alg:line:advance-read-lock}
        \For{$\tx' \in T'$} $\Call{TryTriggerExecution}{\tx'}$ \EndFor \label{alg:line:trigger-execution-freeing-read-locks}
        \EndFor
        \EndFunction

        \Statex
        \Function{HasDependencies}{$\tx$}
        \State $I \gets \Call{HandledObjects}{\tx}$
        \State \Return $\exists \objectid \in I \text{ s.t. } \tx \notin \pendingdb[\objectid][0]$ \label{alg:line:check-next-in-queue}
        \EndFunction

        \Statex
        \Function{MissingObjects}{$\tx$}
        \State $I \gets \Call{HandledObjects}{\tx}$
        \State \Return $ \{ \objectid \text{ s.t. } \objectid \in I \text{ and } \objectsdb[\objectid] = \none \text{ and } \executedIdx < \Call{Index}{\tx}-1 \}$ \label{alg:line:missing-objects}
        \EndFunction

        \Statex
        \Function{HandledObjects}{$\tx$}
        \State \Return $\{ \objectid \text{ s.t. } \objectid \in \tx \text{ and } \Call{Handler}{\objectid}=\self\}$
        \EndFunction

        \Statex
        \Function{AdvanceLock}{$\tx, \objectid$} \label{alg:line:advance-lock}
        \State // Cleanup the pending queue
        \State $(op, T) \gets \pendingdb[\objectid][0]$
        \State $(op, T') \gets (op, l \setminus \tx)$ \label{alg:line:update-pending-queue}
        \State $\pendingdb[\objectid][0] \gets (op, T')$
        \State \Return $T'$
        \EndFunction

        \Statex
        \Function{TryAdvanceExecWatermark}{$\tx$}
        \State $\executed \gets \executed \cup \Call{Index}{\tx}$ \label{alg:line:update-watermark-buffer}
        \While{$(\executedIdx + 1) \in \executed$} $\executedIdx \gets \executedIdx + 1$ \EndWhile \label{alg:line:advance-executed-idx}
        \EndFunction
    \end{algorithmic}
\end{algorithm}

\begin{algorithm}[t]
    \caption{Process $\ready$ (step \four of \Cref{fig:overview})}
    \label{alg:process-ready}
    \algfontsize
    \begin{algorithmic}[1]
        \State $\objectsreceived \gets \{\}$ \Comment{Maps \tx to the object data it references (or $\perp$ if unavailable)}
        \Statex
        \Statex // Called by the \execworkers upon receiving a $\ready$.
        \Procedure{ProcessReady}{$\ready$}
        \State $(\tx, O) \gets \ready$
        \State $\objectsreceived[\tx] \gets \objectsreceived[\tx] \cup O$
        \If{$\len{\objectsreceived[\tx]} \neq \len{\readset{\tx}}+\len{\writeset{\tx}}$} \Return \EndIf \label{alg:line:wait-for-objects}
        \State
        \State $\result \gets (\tx, \emptyset, \emptyset)$ \label{alg:line:empty-result}
        \If{$!\Call{AbortExec}{\tx}$} \label{alg:line:can-execute}
        \State $(O, I) \gets \exec{\tx, \txdb[\tx]}$ \Comment{$O$ to mutate and $I$ to delete} \label{alg:line:execute}
        \For{$w \in  \text{\execworkers}$}
        \State $O_w \gets \{ o \in O \text{ s.t. } \Call{Handler}{o}=w \}$
        \State $I_w \gets \{ \objectid \in I \text{ s.t. } \Call{Handler}{\objectid}=w \}$
        \State $\result \gets (\tx, O_w, I_w)$ \label{alg:line:send-result}
        \EndFor
        \EndIf
        \State $\Call{Send}{w, \result}$ \label{alg:line:send-result-2}
        \EndProcedure

        \Statex
        \Statex // Check whether the execution should proceed.
        \Function{AbortExec}{$\tx$}
        \State \Return $\exists o \in \objectsreceived[\tx] \text{ s.t. } o = \perp$ \label{alg:line:abort-condition}
        \EndFunction
    \end{algorithmic}
\end{algorithm}

\begin{algorithm}[t]
    \caption{Process $\result$ (step \five of \Cref{fig:overview})}
    \label{alg:process-result}
    \algfontsize
    \begin{algorithmic}[1]
        \Statex
        \Statex // Called by the \execworkers upon receiving a $\result$.
        \Procedure{ProcessResult}{$\result$}
        \State $(\tx, O, I) \gets \result$
        \State $\Call{TryAdvanceExecWatermark}{\tx}$ \Comment{Defined in \Cref{alg:core}} \label{alg:line:advance-watermark}
        \State $\Call{UpdateStores}{\tx, O, I}$ \label{alg:line:update-stores}
        \State
        \State // Try execute transactions with missing objects
        \For{$o \in O$}
        \State $\objectid \gets \Call{Id}{o}$
        \For{$\tx \gets \missingdb[\objectid]$} $\Call{TryTriggerExecution}{\tx}$ \EndFor
        \State \delete $\missingdb[\objectid]$ \Comment{Prevent duplicate execution} \label{alg:line:cleanup-missing}
        \EndFor
        \State
        \State // Try executing the next transaction in the queues
        \For{$\objectid \in \tx$}
        \State $T' \gets \Call{AdvanceLock}{\tx, \objectid}$ \label{alg:line:advance-lock-call}
        \For{$\tx' \in T'$} $\Call{TryTriggerExecution}{\tx'}$ \EndFor \label{alg:line:trigger-execution-after-result}
        \EndFor
        \EndProcedure

        \Statex
        \Function{UpdateStores}{$\tx, O, I$}
        \For{$o \in O$} $\objectsdb[\Call{Id}{o}] \gets \object$ \EndFor
        \For{$\objectid \in I$} \delete $\objectsdb[\objectid]$ \EndFor
        \EndFunction
    \end{algorithmic}
\end{algorithm}

\subsection{Running in Constant Memory}
The algorithms described above leverage several temporary in-memory structures that need to be safely cleaned up to make the protocol memory-bound.
The maps \pendingdb and \missingdb are respectively cleaned up as part of normal protocol operations at \Cref{alg:line:update-pending-queue} (empty queues are deleted) and \Cref{alg:line:cleanup-missing} of \Cref{alg:process-result}.
All indices $i' < \loadedIdx$ of the list $\loaded$ (\Cref{alg:process-propose}) can be cleaned after \Cref{alg:line:update-loaded-idx} of \Cref{alg:process-propose} as they are no longer needed.
Similarly, any transactions $\tx$ with index $\Call{Index}{\tx} < \executedIdx$ can be removed from the set $\executed$ (\Cref{alg:core}) after \Cref{alg:line:advance-executed-idx} of \Cref{alg:core}.
Finally, any transaction $\tx$ can be removed from the map $\objectsreceived$ (\Cref{alg:process-ready}) after \Cref{alg:line:wait-for-objects} of \Cref{alg:process-ready}.
\section{Security Proofs} \label{sec:proofs}
We show that \sysname satisfies the properties of \Cref{sec:properties}.

\subsection{Serializability} \label{sec:serializability-proof}
We show that \sysname satisfies the serializability property (\Cref{def:serializability} of \Cref{sec:properties}). Intuitively, this property states that \sysname executes transactions in a way that is equivalent to the sequential execution of the transactions as it comes from consensus (\Cref{def:sequential-schedule}). The argument leverages the following arguments: (i) \sysname builds the pending queues $\pendingdb$ by respecting the transactions dependencies dictated by the consensus protocol (i.e., the sequential schedule), (ii) \sysname accesses objects in the same order as the sequential schedule, and (iii) \sysname executes transactions in the same order as the sequential schedule.

\begin{definition}[Sequential Schedule] \label{def:sequential-schedule}
    A sequential schedule is a sequence of transactions $[\tx_1,\dots,\tx_n]$ where each transaction $\tx_i$ is executed after $\tx_{i-1}$.
\end{definition}

\begin{definition}[Conflicting Transactions] \label{def:conflicting-transaction}
    Two transactions $\tx_i$ and $\tx_j$ conflict on some object \objectid if both $\tx_i$ and $\tx_j$ reference \objectid in their read or write set and at least one of $\tx_i$ or $\tx_j$ references \objectid in its write set.
\end{definition}

\para{Pending queues building}
We start by arguing point (i), stating that \Cref{alg:process-propose} builds the pending queues $\pendingdb$ by respecting the transaction dependencies dictated by the consensus protocol (i.e., the sequential schedule).

\begin{lemma}[Sequential Batch Processing] \label{lm:process-propose-order}
    \sys processes the batch with index $\batch_j$ after processing the batch with index $\batch_i$ if $j>i$.
\end{lemma}
\begin{proof}
    Let's assume by contradiction that \Cref{alg:process-propose} processes the $\propose$ referencing transactions of the batch with index $\batch_j$ before processing the $\propose$ referencing transactions of the batch with index $\batch_i$ while $j>i$.
    This means that \Cref{alg:process-propose} processes  $\batch_j$ at \Cref{alg:line:process-proposed-txs} before processing $\batch_i$ at \Cref{alg:line:process-proposed-txs}. However, the check of \Cref{alg:process-propose} at \Cref{alg:line:wait-for-batches} ensures that $\batch_j$ can only be processed after all the batches with indices $k \in [0,\dots,j[$. Since $j>i$, it follows that $i \in [0,\dots,j[$, and thus $\batch_j$ can only be processed after $\batch_i$. Hence a contradiction.
\end{proof}

\begin{lemma}[Transactions Order in Queues] \label{lm:queue-order}
    Let's assume two transactions $\tx_j, \tx_i$ such that $j>i$ conflict on the same $\objectid$; $\tx_j$ is placed in the queue $\pendingdb[\objectid]$ after $\tx_i$.
\end{lemma}
\begin{proof}
    We first observe that if two transactions $\tx_j$ and $\tx_i$ are conflicting on object $\objectid$ then they are placed in the same queue $\pendingdb[\objectid]$. Indeed, both $\tx_i$ and $\tx_j$ are embedded in a $\propose$ by \Cref{alg:process-batch}. They are then placed in the queue $\pendingdb[\objectid]$ by \Cref{alg:process-propose} at \Cref{alg:line:schedule-write} (if they reference $\objectid$ in their write set) or \Cref{alg:line:schedule-read} (if they reference $\objectid$ in their read set).

    We are thus left to prove that $\tx_j$ is placed in $\pendingdb[\objectid]$ after $\tx_i$. Since $j>i$ we distinguish two cases: (i) both $\tx_j$ and $\tx_i$ are part of the same batch with index $\batchidx$ and (ii) $\tx_j$ and $\tx_i$ are part of different batches with indices $\batch_j$ and $\batch_i$ respectively.
    In the first case (i), $\tx_j$ and $\tx_i$ are referenced in the same $\propose$ by \Cref{alg:process-batch} at \Cref{alg:line:propose-creation-1} and \Cref{alg:line:propose-creation-2} but respecting the order $j > i$. As a result, $\tx_j$ is processed after $\tx_i$ by the loop \Cref{alg:line:process-proposed-txs}, and placed in the queue $\pendingdb[\objectid]$ (at \Cref{alg:line:schedule-write} or \Cref{alg:line:schedule-read}) after $\tx_i$.
    In the second case (ii), $\tx_j$ and $\tx_i$ are referenced in different $\propose$ by \Cref{alg:process-batch} at \Cref{alg:line:propose-creation-2} but \Cref{lm:process-propose-order} ensures that the $\propose$ referencing transactions of $\batch_j$ is processed after the $\propose$ referencing transactions of $\batch_i$. As a result, $\tx_j$ is placed in the queue $\pendingdb[\objectid]$ after $\tx_i$ (at \Cref{alg:line:schedule-write} or \Cref{alg:line:schedule-read}).
\end{proof}

\para{Sequential objects access}
We now argue point (ii), namely that \Cref{alg:core} accesses objects in the same order as the sequential schedule.

\begin{lemma}[Unlock after Access] \label{lm:remove-after-access}
    If a transaction $T$ is placed in a queue $\pendingdb[\objectid]$, it can only be removed from that queue after accessing $\objectsdb[\objectid]$.
\end{lemma}
\begin{proof}
    We argue this lemma by construction of the algorithms of \sysname. Transaction $T$ accesses $\objectsdb[\objectid]$ only at \Cref{alg:line:load-objects} (\Cref{alg:core}) and can only be removed from $\pendingdb[\objectid]$ following a call to $\Call{AdvanceLock}{T, \objectid}$. This call can occur only in two places.
    It can first occur (i) at \Cref{alg:line:advance-read-lock} of \Cref{alg:core} which happens after the access to $\objectsdb[\objectid]$ (\Cref{alg:line:load-objects} of the same algorithm).
    It can then occur (ii) at \Cref{alg:line:advance-lock-call} of \Cref{alg:process-result} which can only be triggered upon receiving a $\result$ referencing $T$, which in turn can only be created after creating a $\ready$ embedding $T$. However, creating the latter message only occurs at \Cref{alg:line:send-ready} of \Cref{alg:core}, thus after accessing $\objectsdb[\objectid]$ (\Cref{alg:line:load-objects} of that same algorithm).
\end{proof}

\begin{lemma}[Sequential Object Access] \label{lm:access-order}
    If a transaction $\tx_j$ is placed in a queue $\pendingdb[\objectid]$ after a transaction $\tx_i$, then $\tx_j$ accesses $\objectid$ after $\tx_i$.
\end{lemma}
\begin{proof}
    Let's assume that $\tx_j$ and $\tx_i$ are respectively placed at positions $j'$ and $i'$ of the queue $\pendingdb[\objectid]$, with $j'>i'$. Let's assume by contradiction that $\tx_j$ accesses $\objectid$ before $\tx_i$.
    Access to $\objectid$ is only performed by \Cref{alg:core} at \Cref{alg:line:load-objects} after successfully passing the `dependencies' check at \Cref{alg:line:check-dependencies}. \Cref{lm:remove-after-access} thus ensures that $\tx_i$ is still in $\pendingdb[\objectid]$ when the call to  $\Call{HashDependecnies}{\tx_j}$ at \Cref{alg:line:check-dependencies} returns \textbf{False}.
    This is however a direct contradiction of the check at \Cref{alg:line:check-next-in-queue} which ensures that $\Call{HasDependecnies}{\tx_j}$ returns \textbf{False} only if $\tx_j$ is in the $\pendingdb[\objectid]$ at position $j'=0$. However, since $\tx_i$ is still in $\pendingdb[\objectid]$, it follows that $0 \leq i' < j'$, thus a contradiction.
\end{proof}

\para{Sequential transaction execution}
We finally argue point (iii), namely that \Cref{alg:process-ready} executes transactions in the same order as the sequential schedule.

\begin{lemma}[Execution after Object Access] \label{lm:access-then-execute}
    If a transaction $T$ references $\objectid$ in its read or write set, it can only be executed after accessing $\objectsdb[\objectid]$.
\end{lemma}
\begin{proof}
    We argue this lemma by construction of \Cref{alg:process-ready}. Transactions are executed only at \Cref{alg:line:execute} of \Cref{alg:process-ready} and this algorithm is only triggered upon receiving a $\ready$. However, creating the latter message only occurs at \Cref{alg:line:send-ready} of \Cref{alg:core}, thus after accessing $\objectsdb[\objectid]$ (\Cref{alg:line:load-objects} of that same algorithm).
\end{proof}

\begin{theorem}[Serializability] \label{th:serializability}
    If a correct validator executes the sequence of transactions $[\tx_1,\dots,\tx_n]$, it holds the same object state $S$ as if the transactions were executed sequentially.
\end{theorem}
\begin{proof}
    Consider some execution $E$ and let $G = (V,E)$ be $E$'s conflict graph. Each transaction is a vertex in $V$, and there is a directed edge $\tx_i \rightarrow \tx_j$ if (1) $\tx_i$ and $\tx_j$ have a conflict on some object $\objectid$ and (2) $\tx_j$ accesses $\objectid$ after $\tx_i$ accesses $\objectid$. It is sufficient to show that there are no schedules where $\tx_j$ is executed before $\tx_i$ to prove serializability.
    Let's assume by contradiction that there is a schedule where $\tx_j$ is executed before $\tx_i$, where $j > i$. Since $\tx_j$ and $\tx_i$ conflict on object $\objectid$, \Cref{lm:queue-order} ensures that $\tx_j$ is placed in the queue $\pendingdb[\objectid]$ after $\tx_i$. \Cref{lm:access-order} then guarantees that $\tx_j$'s access to $\objectid$ occurs after $\tx_i$'s access on $\objectid$. However \Cref{lm:access-then-execute} ensures that $\tx_j$'s execution can only happen after accessing $\objectid$. It is then impossible to execute $\tx_j$ before $\tx_i$, hence a contradiction. Since $\objectid$ was chosen arbitrarily, the same reasoning applies to all objects on which $\tx_i$ and $\tx_j$ conflict.
\end{proof}

\subsection{Determinism} \label{sec:determinism-proof}
We show that \sysname satisfies the determinism property (\Cref{def:determinism} of \Cref{sec:properties}). Intuitively, this property ensures that all correct validators have the same object state after executing the same sequence of transactions. The proof follows from the following arguments: (i) all correct \sysname validators build the same dependency graph given the same input sequence of transaction, (ii) individual transaction execution is a deterministic process (\Cref{as:deterministic-execution}), (iii) transactions explicitly reference their entire read and write set (\Cref{as:explicit-read-write-set}), and (iv) all validators executing the same transactions obtain the same state.

\begin{assumption}[Deterministic Individual Execution] \label{as:deterministic-execution}
    Given an input transaction $\tx$ and objects $O$, all calls to $\exec{\tx, O}$ (\Cref{alg:line:execute} of \Cref{alg:process-ready}) return the same output.
\end{assumption}

\begin{assumption}[Explicit Read and Write Set] \label{as:explicit-read-write-set}
    Each transaction $\tx$ explicitly references all the objects of its read and write set. That is, the complete read and write set of $\tx$ can be determined by locally inspecting $\tx$ without the need for external context.
\end{assumption}

\Cref{as:deterministic-execution} is fulfilled by most blockchain execution environments such as the EVM~\cite{evm}, the SVM~\cite{solana-vm}, and both the MoveVM~\cite{move-vm} (used by the Aptos blockchain~\cite{aptos}) and the Sui MoveVM~\cite{sui-move} (used by the Sui blockchain~\cite{sui-lutris}). All execution engines except the Sui MoveVM also fulfill \Cref{as:explicit-read-write-set} (\Cref{sec:dynamic} and \Cref{sec:child-details} remove this assumption to make \sysname compatible with Sui).

We rely on the following lemmas to prove determinism in \Cref{th:determinism}.

\begin{lemma} \label{lm:same-queues}
    Given the same sequence of transactions $[\tx_1,\dots,\tx_n]$, all correct validators build the same  execution schedule (that is, they build same queues $\pendingdb$).
\end{lemma}
\begin{proof}
    We argue this property by construction of \Cref{alg:process-batch} and \Cref{alg:process-propose}.
    Since all validators receive the same input sequence $[\tx_1,\dots,\tx_n]$ and \Cref{alg:process-batch} respects the order of transactions (\Cref{alg:line:propose-creation-1}), all correct validators create the same sequence of $\propose$. \Cref{lm:process-propose-order} then ensures that \Cref{alg:process-propose} processes each $\propose$ respecting the transaction order. Finally, \Cref{lm:queue-order} ensures that all correct validators place the transactions in the same order in the queues. Since this process is deterministic, all correct validators build the same queues $\pendingdb$.
\end{proof}

\begin{lemma} \label{lm:same-state}
    No two correct validators creating the same $\result$ (\Cref{alg:line:send-result} of \Cref{alg:process-ready}) obtain a different object state $\objectsdb$.
\end{lemma}
\begin{proof}
    We argue this property by construction of \Cref{alg:process-result} and by assuming that the communication channel between all \execworkers of each validator preserves the order of messages.\footnote{Our implementation (\Cref{sec:implementation}) satisfies this assumption by implementing all communication through TCP.}
    Once a validator creates a $\result$ (\Cref{alg:line:send-result} of \Cref{alg:process-ready}), it is processed by \Cref{alg:process-result}. This algorithm first calls $\Call{TryAdvanceExecWatermark}{\cdot}$ (\Cref{alg:line:advance-watermark}) which does not alter the object state $\objectsdb$ nor the data carried by the $\result$, and then deterministically updates the object state $\objectsdb$ (\Cref{alg:line:update-stores}) based exclusively on the content of the $\result$.
    As a result, all correct validators obtain the same object state
\end{proof}

\begin{theorem}[\sysname Determinism] \label{th:determinism}
    No two correct validators that executed the same sequence of transactions $[\tx_1, \dots, \tx_n]$ have different stores $\objectsdb$.
\end{theorem}
\begin{proof}
    We argue this property by induction. Assuming the sequence of conflicting transactions $[\tx_1,\dots,\tx_{n-1}]$ for which this property holds, we consider transaction $\tx_n$.
    \Cref{lm:same-queues} ensures that all correct validators build the same execution schedule and thus all correct validators execute conflicting transactions in the same order.
    After scheduling (\Cref{alg:process-propose}), all correct validators create a $\ready$ referencing $\tx_n$ and the set of objects $O$ (\Cref{alg:line:send-ready} of \Cref{alg:core}). Since all validators have the same conflict schedule and the application of the inductive argument ensures that all settled dependencies of $\tx_n$ led to the same state $\objectsdb$ across validators, all correct validators load the same set of objects $O$ and thus create the same $\ready$.
    As a result, all correct validators run \Cref{alg:process-ready} with the same input and thus execute the same sequence of transactions. By construction of \Cref{alg:process-ready} and \Cref{as:explicit-read-write-set}, they all call calls to $\exec{\tx, O}$ (\Cref{alg:line:execute} of \Cref{alg:process-ready}) with the same inputs $\tx$ and $O$. Given that \Cref{as:deterministic-execution} ensures that all calls to $\exec{\tx, O}$ are deterministic, all correct validators thus create the same $\result$ (\Cref{alg:line:send-result}). Finally, \Cref{lm:same-state} ensures that all validators creating the same $\result$ obtain the same object state $\objectsdb$.
    The inductive base is argued by construction: all correct validators start with the same object state $\objectsdb$, and thus create the same $\ready$ (\Cref{alg:line:send-ready} of \Cref{alg:core}) and $\result$ (\Cref{alg:line:send-result} of \Cref{alg:process-ready}) upon executing the first transaction $\tx_1$, which leads to the same state update across correct validators.
\end{proof}

\subsection{Liveness}
We show that \sysname satisfies the liveness property (\Cref{def:liveness} of \Cref{sec:model}). Intuitively, this property guarantees that valid transactions (\Cref{def:valid-transaction}) are eventually executed. The proof argues that (i) all transactions are eventually processed (\Cref{def:processed-transaction}), and (ii) among those transactions, valid ones are not aborted.

\begin{definition}[Valid Transaction] \label{def:valid-transaction}
    A transaction $T$ with index $idx = \Call{Index}{T}$ is valid if all objects referenced by its read and write set are created by a transaction $T'$ with index $idx' < idx$.
\end{definition}

\begin{definition}[Processed Transaction] \label{def:processed-transaction}
    A transaction $T$ is said \emph{processed} when it is either executed or aborted and the object state $\objectsdb$ is updated accordingly.
\end{definition}

\para{Eventual transaction processing}
We start by arguing point (i), that is all transactions are eventually processed. This argument relies on several preliminary lemmas leading \Cref{lm:eventual-processing}.

\begin{lemma} \label{lm:no-deadlock}
    The \sysname scheduling process is deadlock-free (no circular dependencies).
\end{lemma}
\begin{proof}
    Consider some execution $E$ and let $G = (V,E)$ be $E$'s conflict graph. Each transaction is a vertex in $V$, and there is a directed edge $\tx_i \rightarrow \tx_j$ if (1) $\tx_i$ and $\tx_j$ have a conflict on some object \objectid and (2) $\tx_j$ accesses $\objectid$ after $\tx_i$ accesses $\objectid$. It is sufficient to show that $G$ contains no cycles to prove liveness. That is, it is sufficient to show that $G$ contains no edges $\tx_j \rightarrow \tx_i$, where $j > i$.
    Let's assume by contradiction that $G$ has an edge $\tx_j \rightarrow \tx_i$, where $j > i$. Then, by rule (1) of the construction of $G$, $\tx_j$ and $\tx_i$ must conflict on some object $\objectid$.
    \Cref{lm:queue-order} ensures that $\tx_j$ is placed in the queue in $\pendingdb[\objectid]$ after $\tx_i$ . \Cref{lm:access-order} then guarantees that $\tx_j$'s access to $\objectid$ occurs after $Ti$'s access on $\objectid$. It is then impossible for $G$ to contain an edge $\tx_j \rightarrow \tx_i$ as this violates rule (2) of the construction of $G$, hence a contradiction. Since $\objectid$ was chosen arbitrarily, the same reasoning applies to all objects on which $\tx_i$ and $\tx_j$ conflict.
\end{proof}

\begin{lemma} \label{lm:clear-queues}
    If a transaction $T$ is processed (\Cref{def:processed-transaction}), it is eventually removed from all queues $\pendingdb[\objectid]$ where $\objectid$ is referenced by the read or write set of $T$.
\end{lemma}
\begin{proof}
    We argue this lemma by construction of \Cref{alg:core} and \Cref{alg:process-result}. Transaction $T$ is  removed from all queues $\pendingdb[\objectid]$ upon calling $\Call{AdvanceLock}{T,\objectid}$. This call can occur in two places.
    (i) The first call occurs at \Cref{alg:line:advance-read-lock} of \Cref{alg:core} to effectively release read locks. This call occurs right after creating a $\ready$ referencing $T$ (\Cref{alg:line:send-ready}), a necessary step to trigger \Cref{alg:process-ready} and thus process $T$.
    (ii) The second call occurs at \Cref{alg:line:advance-lock-call} of \Cref{alg:process-result} to effectively release write locks. This call occurs right after updating $\objectsdb[\objectid]$ and thus terminating the processing of $T$.
\end{proof}

\begin{lemma} \label{lm:advance-exec-watermark}
    If the sequence of transactions $[\tx_1, \dots, \tx_n]$ is processed (\Cref{def:processed-transaction}), the watermark $\executedIdx$ (\Cref{alg:line:executed-idx} of \Cref{alg:core}) is advanced to $n$.
\end{lemma}
\begin{proof}
    The processing of $T$ involves updating the object state $\objectsdb$ (\Cref{alg:line:update-stores} of \Cref{alg:process-result}). However, the watermark $\executedIdx$ is only updated upon calling $$\Call{TryAdvanceExecWatermark}{\tx_i}$$ at \Cref{alg:line:advance-watermark} of this same algorithm. Thus, by construction, the buffer $\executed$ contains every processed transaction $\tx_i$ (\Cref{alg:line:update-watermark-buffer} of \Cref{alg:core}), and once the sequence $[\tx_1, \dots, \tx_n]$ is processed, the watermark $\executedIdx$ is advanced to $$\executedIdx = \max\{\Call{Index}{\tx_1},\dots,\Call{Index}{\tx_n}\} = n$$ (\Cref{alg:line:advance-executed-idx} of \Cref{alg:core}).
\end{proof}

\begin{lemma} \label{lm:process-if-ready}
    A transaction $T$ is eventually processed (\Cref{def:processed-transaction}) if it has neither missing dependencies nor missing objects that could be created by earlier transactions. That is, $T$ is eventually processed if \Cref{alg:core} creates a $\ready$ referencing $T$.
\end{lemma}
\begin{proof}
    We argue this lemma by construction of \Cref{alg:process-ready} and \Cref{alg:process-result}. \Cref{alg:process-ready} receives a $\ready$ from all \execworkers and the check \Cref{alg:line:wait-for-objects} of \Cref{alg:process-ready} passes. Transaction $T$ is then either executed (if the check \Cref{alg:line:abort-condition} passes) or aborted (if the check \Cref{alg:line:abort-condition} fails), and \Cref{alg:process-ready} creates a $\result$ referencing $T$ (\Cref{alg:line:send-result}). \Cref{alg:process-result} then receives this $\result$ and accordingly updates its object $\objectsdb$ (after an infallible call to $\Call{TryAdvanceExecWatermark}{T}$ \Cref{alg:line:advance-executed-idx}).
\end{proof}

\begin{lemma}[Eventual Transaction Processing] \label{lm:eventual-processing}
    All correct validators receiving the sequence of transactions $[\tx_1, \dots, \tx_n]$ eventually process (\Cref{def:processed-transaction}) all transactions $\tx_1, \dots, \tx_n$.
\end{lemma}
\begin{proof}
    \Cref{lm:no-deadlock} ensures that the transaction scheduling process is deadlock-free (no circular dependencies) and \sysname thus triggers their execution (\Cref{alg:line:try-trigger-execution} of \Cref{alg:core}).
    We are then left to prove that these scheduled transactions are processed (\Cref{def:processed-transaction}). Since \Cref{th:serializability} ensures the \sysname schedule is equivalent to a sequential schedule, we prove liveness of the sequential schedule. We argue the lemma's statement by induction. Assuming the sequence of transactions $[\tx_1,\dots,\tx_{n-1}]$ for which this statement holds, we consider transaction $\tx_n$.
    Assuming $\tx_{n-1}$ is processed and a sequential schedule, all transactions $\tx_i$ with $i < n-1$ are also processed. \Cref{lm:clear-queues} thus ensures these transactions are removed from all queues $\pendingdb[\cdot]$. As a result, when triggering the execution of $\tx_n$ (\Cref{alg:line:try-trigger-execution} of \Cref{alg:core}), the check $\Call{HasDependencies}{\tx_n}$ (\Cref{alg:line:check-dependencies} of \Cref{alg:core}) returns \textbf{False} (since $\forall \objectid \in \readset{\tx_n}\cup\writeset{\tx_n}: \pendingdb[\objectid][0] = \tx_n$).
    Furthermore, since all transactions $\tx_i$ with $i \leq n-1$ are already processed, \Cref{lm:advance-exec-watermark} ensures that the watermark $j=n-1$ (\Cref{alg:line:executed-idx} of \Cref{alg:core}) and thus $\Call{MissingObjects}{\tx_n}$ returns $\emptyset$.
    Finally, \Cref{alg:core} creates a $\ready$ referring $\tx_n$ and thus \Cref{lm:process-if-ready} ensures $\tx_n$ is eventually processed.
    We argue the inductive base by observing that the first transaction $\tx_1$ has no dependency (by definition); thus both checks $\Call{HasDependencies}{\tx_1}$ and $\Call{MissingObjects}{\tx_1}$ pass (respectively at \Cref{alg:line:check-dependencies} and \Cref{alg:line:check-missing-objects} of \Cref{alg:core}); and \Cref{lm:process-if-ready} then ensures $\tx_1$ is processed.
\end{proof}

\para{Valid transaction execution}
We now argue point (ii), that valid transactions are not aborted. This argument relies on several preliminary lemmas leading \Cref{lm:execute-valid-transaction}.

\begin{lemma} \label{lm:remove-after-write}
    If a transaction $T$ references an object $\objectid$ in its write set, it is only removed from the queue $\pendingdb[\objectid]$ after it is processed (\Cref{def:processed-transaction}).
\end{lemma}
\begin{proof}
    We argue this lemma by construction: Transaction $T$ is removed from $\pendingdb[\objectid]$ only upon a call to $\Call{AdvanceLock}{T, \objectid}$. However, since $\objectid$ is referenced by the write set of $T$ (rather than its read set), this function is called over $\objectid$ only at \Cref{alg:line:advance-lock-call} of \Cref{alg:process-result}. This call is thus after \Cref{alg:process-result} updates $\objectsdb[\objectid]$ (at \Cref{alg:line:update-stores}) and thus after the transaction is processed.
\end{proof}

\begin{lemma} \label{lm:inputs-always-available}
    If a transaction $T$ is valid, the call to $\objectsdb[\objectid]$ (\Cref{alg:line:load-objects} of \Cref{alg:line:advance-executed-idx}) never returns $\perp$, for any $\objectid$ referenced by the read or write set of $T$.
\end{lemma}
\begin{proof}
    Let's assume by contradiction that there exists a $\objectid$ referenced by the read or write set of a valid transaction $T$ where the call to $\objectsdb[\objectid]$ (\Cref{alg:line:load-objects} of \Cref{alg:line:advance-executed-idx}) returns $\perp$.
    Since $T$ is valid, it means that the object $\objectid$ is created by a conflicting transaction $T'$ with index $idx' < \Call{Index}{T}$ that has not yet been processed (\Cref{def:processed-transaction}). In which case, \Cref{lm:access-order} states that both $T$ and $T'$ are placed in the same queue $\pendingdb[\objectid]$ and \Cref{lm:remove-after-write} states that $T'$ is still present in the queue $\pendingdb[\objectid]$. This is however a direct contradiction of check $\Call{HasDependencies}{T}$ ensuring that $T$ does not access $\objectsdb[\objectid]$ until it is at the head of the queue $\pendingdb[\objectid]$ (\Cref{alg:line:load-objects} of \Cref{alg:core}).
\end{proof}

\begin{lemma} \label{lm:execute-valid-transaction}
    A valid transaction $T$ it is never aborted; that is the call $\Call{AbortExec}{T}$ (\Cref{alg:line:can-execute} of \Cref{alg:process-ready}) returns \textbf{False}.
\end{lemma}
\begin{proof}
    Let's assume by contradiction that $\Call{AbortExec}{T}$ returns \textbf{True} while $T$ is valid. This means that the check at \Cref{alg:line:abort-condition} of \Cref{alg:process-ready} found at least one missing object ($\perp$) referenced in the read or write set of $T$, and thus that \Cref{alg:process-ready} received at least one $\ready$ message referring $\perp$ instead of an object data. However, \Cref{lm:inputs-always-available} ensures that the call to $\objectsdb[\objectid]$ (\Cref{alg:line:load-objects} of \Cref{alg:core}) never returns $\perp$ if $\objectid$ is referenced by the read or write set of a valid transaction, hence a contradiction.
\end{proof}

\para{Liveness proof}
We finally combine \Cref{lm:eventual-processing} and \Cref{lm:execute-valid-transaction} to prove liveness.

\begin{theorem}[Liveness] \label{th:liveness}
    All correct validators receiving the sequence of transactions $[\tx_1, \dots, \tx_n]$ eventually execute all the valid transactions of the sequence.
\end{theorem}
\begin{proof}
    \Cref{lm:eventual-processing} ensures that all correct validators eventually process (\Cref{def:processed-transaction}) all transactions $\tx_1, \dots, \tx_n$. \Cref{lm:execute-valid-transaction} ensures that valid transactions are never aborted and thus executed.
\end{proof}
\section{Crash Fault Tolerance} \label{sec:cft}

\subsection{Internal Replication}

For the replication protocol, \execworkers maintains the following network connections: (i) a constant set of \emph{\peers}, containing the identifier for every worker in its cluster. Workers in each cluster have the same \peers set; (ii) a dynamic set of \readto, containing the additional identifiers with whom the worker is temporarily serving reads to; and (iii) a dynamic set of \readfrom, containing the additional identifiers with whom the worker is receiving reads from. The protocol maintains that \readfrom and \readto relations are symmetric --- Worker $a$ is in worker $b$'s \readfrom set if and only if worker $b$ is in $a$'s \readto set. Finally, we assume the use of an eventually strong failure detector~\cite{chandra96failure}.

\subsection{Normal Operation}

\begin{figure}[t]
    \begin{center}
        \includegraphics[width=0.5\columnwidth]{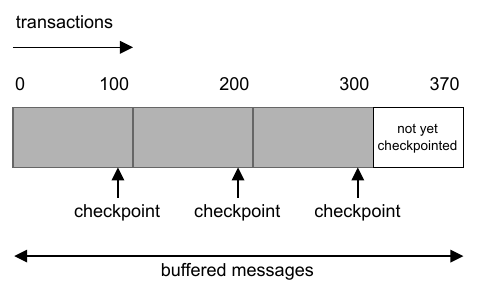}
        \caption{Example of a snapshot of the local state at an \execworker. It checkpoints its \objectsdb store after every 100 transactions and keeps a buffer of outgoing \ready.}
        \label{fig:ew-no-gc}
    \end{center}
\end{figure}

\para{Dealing with finite memory}
The number of checkpoints and buffered messages held by an \execworker cannot grow indefinitely. Hence, we introduce a garbage collection mechanism that deprecates old checkpoints.
When an \execworker completes a checkpoint, it broadcasts a message
$$\checkpointed \gets (shard, \txidx)$$
to every other \execworker in \emph{every} shard, indicating that an \execworker of shard $shard$ successfully persisted a checkpoint immediately after executing \txidx.

An \execworker deems a checkpoint at \txidx as \emph{stable} after receiving a quorum of $f_e+1$ \checkpointed\footnote{Quorum sizes can be varied to optimize between normal case disk-usage and recoverability during failure, similar to Flexible Paxos~\cite{howard2016flexible}.} from each shard.
When a worker learns that a checkpoint is stable, it deletes all checkpoints and buffered messages prior to that checkpoint. This illustrated in \Cref{fig:ew-garbage}.

\begin{figure}[t]
    \begin{center}
        \includegraphics[width=0.5\columnwidth]{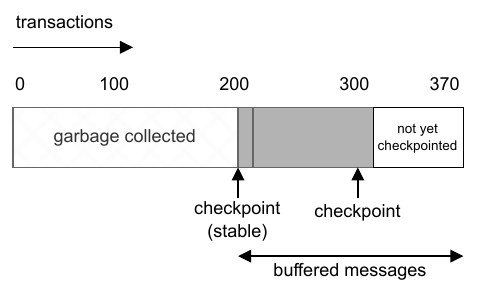}
        \caption{Example of a snapshot of the local state at an \execworker. It has received a quorum of \checkpointed messages from each shard for the transaction at checkpoint 1. Hence, checkpoint 1 is stable, and the worker safely deletes checkpoints and buffers before it.}
        \label{fig:ew-garbage}
    \end{center}
\end{figure}

\para{Bounding strategy}
To avoid exhausting resources, each \execworker also holds a bounded number $c$ of checkpoints at any time. This number dictates how far \execworkers are allowed to diverge in terms of their rate of execution; the fastest cluster can be ahead of the slowest cluster in its quorum by up to $c-1$ checkpoints.\footnote{Note that \execworkers within a cluster are always tightly coupled due to the quorum definition, and can never be apart by 1 or more checkpoints}

A worker pauses processing when creating a new checkpoint, which will exceed $c$. This may be a symptom of failures in the system, e.g., many slow or failed workers or network issues. Hence, pausing provides backpressure to fast replicas in order for stragglers to catch up.
\Cref{fig:ew-bounding} illustrates this mechanism in a system with three clusters and $c=2$. Each worker of each cluster holds a stable checkpoint at boundary 1. Clusters 1 and 2 are slow and have yet to reach checkpoint boundary 2. Cluster 3 is fast and hence may execute beyond checkpoint boundary 2, while maintaining a second (non-stable) checkpoint at boundary 2.
However, because workers are limited to storing two checkpoints, workers in cluster 3 are blocked from executing past boundary 3 before (i) their checkpoint at boundary 2 is established as stable and (ii) their checkpoint 1 is garbage-collected.

By default, we set $c=2$ as a good trade-off between performance and storage/memory costs. With a limit of two checkpoints, a fast cluster can execute past a second checkpoint without waiting for a quorum. As such, different clusters are allowed to progress at different speeds without blocking, as long as they stay within one checkpoint. The system then progresses, most of the time, at the speed of the fastest cluster.

\begin{figure}[t]
    \begin{center}
        \includegraphics[width=0.5\textwidth]{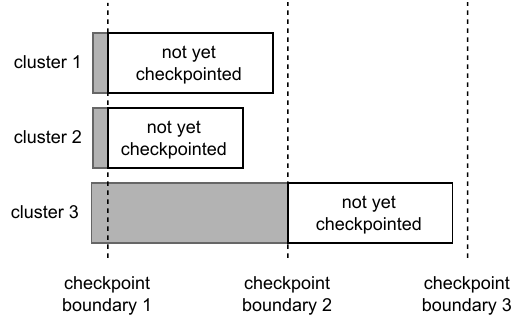}
        \caption{Consider a system with three clusters, and with each \execworker allowed to hold up to $c=2$ checkpoints. This snapshot of the clusters' progress shows that \execworkers in each cluster have a stable checkpoint at boundary 1. Cluster 3 is fast and hence may execute beyond checkpoint boundary 2 while maintaining a second (non-stable) checkpoint at boundary 2. However, it is not permitted to create a checkpoint at boundary 3 or execute past it before it learns that boundary 2 is stable. }
        \label{fig:ew-bounding}
    \end{center}
\end{figure}

\subsection{Failure Recovery} \label{sec:failure-recovery}

\para{Recovery through reconfiguration}
When an \execworker fails, other workers in its cluster or \readto set may not be able to execute transactions, as they may no longer be served the necessary reads. To restore transaction processing, these workers trigger the \emph{reconfiguration} illustrated by \Cref{fig:reconfig}. In order to detect these crashes, we assume the existence of a failure detector~\cite{chandra96failure} with strong completeness. Ideally, we would use an eventually perfect failure detector, but an eventually strong one suffices for liveness (but might cause worse load-balancing on some \execworkers).

The crux of recovery is as follows: When a worker detects a failure, it tries to find another worker to get the reads previously managed by the failed worker. In the normal case, this takes two round trips: one trip to find an appropriate \readfrom member and another to establish the relationship with that new member. Meanwhile, all other clusters except the one with the failed worker operate as normal. Hence, there is no loss of throughput when failures are within the tolerated threshold. 
\ifpublish
We defer the details of the recovery algorithm to Appendix~\ref{sec:detailed-recovery}.
\else
The full version of our paper presents the full details of the recovery algorithm.
\fi

\begin{figure}[t]
    \begin{center}
        \includegraphics[width=.4\textwidth]{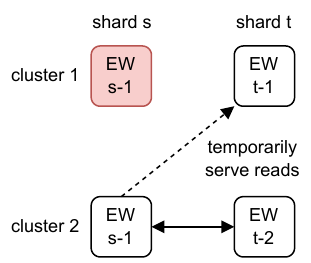}
        \caption{
            Failure recovery. Suppose \ew{s}{1} crashes. As a result, \ew{t}{1} cannot receive reads of shard $s$ from \ew{s}{1}. After \ew{t}{1} performs recovery, it establishes a new read relationship with another replica of shard $s$, in this case \ew{s}{2}, as illustrated by the dashed arrow. \ew{t}{1} is now in \ew{s}{2}'s \readto set, and correspondingly, \ew{s}{2} is in \ew{t}{1}'s \readfrom set. Otherwise, cluster $2$ operates as usual, as represented by the solid arrow.}
        \label{fig:reconfig}
    \end{center}
\end{figure}

\para{Recovery through checkpoints synchronization}
The recovery through reconfiguration may fail when the recovering worker finds itself slow after the first round trip. It then needs to perform a \emph{checkpoint synchronization} procedure before retrying recovery. This synchronization is necessary as there may no longer be clusters with sufficiently old \textit{buffers} for the recovering worker to continue execution through the normal recovery procedure.

The gist of the synchronization procedure is as follows:. (i) The worker instructs every peer to perform \emph{synchronization}; (ii) any node in the slow worker's \readto set is instructed to itself perform recovery; (iii) the worker then downloads the latest checkpoints from another replica and waits for every peer to sync to the same state; and (iv) the worker finds replacements for any missing members in its cluster by running the \emph{reconfiguration} again. 
\ifpublish
We defer the details of the synchronization algorithm to Appendix~\ref{sec:detailed-recovery}.
\else
The full version of our paper~\cite{fullversion} provides the details of the synchronization algorithm.
\fi

The synchronization protocol's pivotal aspect lies in its ability to activate synchronization and recovery across all clusters that transitively rely on a slow worker for reads. This occurs because the clusters, which were dependent on the slow worker for their reads, may exhibit lag as well. If the slow worker were to unilaterally fast-forward its state, these clusters could potentially lose liveness.

\para{Disaster recovery}
In case of a disaster that affects all \execworkers of a cluster and the threat model of \Cref{sec:threat-model} doesn't hold, the cluster can be recovered by booting a new cluster with the same \peers set. The new cluster will then be able to recover the state of the cluster from the other validators of the blockchain. This is possible because the system state is replicated across multiple validators, of which at least half are honest. The new cluster can then download the latest stable checkpoint from the other validators and use it to perform a recovery through checkpoints. This recovery is slow as it requires WAN communications, but it is only used in extreme circumstances, and its existence allows \sysname to be reasonably configured with low replication factors (e.g., $f_e=1$).

\section{Detailed Recovery Protocol}
\label{sec:detailed-recovery}
This section completes \Cref{sec:cft} by providing the algorithms allowing \execworkers to recover from crash-faults and proving the security of \sysname in this setting.

\subsection{Recovery Algorithms} \label{sec:recovery-algorithms}


\para{Recovery Protocol} Suppose \execworker $x$ crashes. Any non-faulty worker $e$ detecting this failure deletes $x$ from its \readfrom and \readto sets. If $x$ is $e$'s peer, or a member of its \readfrom set, $e$ may no longer be served reads from $x$'s shard. In this case, $e$ calls \Call{Recover}{$x$}, listed in \Cref{alg:recovery}.

\begin{algorithm}[t]
    \caption{Process \checkpointed}
    \algfontsize
    \label{alg:garbage-collection}
    \begin{algorithmic}[1]
        \Statex
        \State $\checkpoints \gets \{\}$  \Comment{Maps \txidx to checkpoint blob}
        \State $\sentmessages \gets \{\}$ \Comment{Maps \txidx to set of sent \ready}
        \State $\checkpointsnotification \gets 0$ \Comment{Counts the received \checkpointed}
        \Statex
        \Procedure{ProcessCheckpointed}{$\checkpointed$}
        \State $(shard, \txidx) \gets \checkpointed$
        \State $\checkpointsnotification[(shard, \txidx)] \gets \checkpointsnotification[(shard, \txidx)] + 1$
        \If{$\checkpointsnotification[(shard, \txidx)] \geq f_e+1$} \label{alg:line:delete-condition}
        \For{$i \gets 0; \; i < \txidx; \; i \gets i+1$}
        \State \delete $\checkpoints[\checkpointsIdx]$
        \State \delete $\sentmessages[\checkpointsIdx]$
        \State \delete $\checkpointsnotification[(*, \txidx)]$
        \EndFor
        \EndIf
        \EndProcedure
    \end{algorithmic}
\end{algorithm}

\begin{algorithm}[t]
    \caption{Recovery procedure}
    \label{alg:recovery}
    \algfontsize
    \begin{algorithmic}[1]
        \Statex // Global states
        \State $\readto, \readfrom, \peers$
        \State \textbf{suspected}  \Comment{set of suspected workers, updated by failure detector}
        \State \textbf{curr-txidx} \Comment{highest txn that is locally executed and persisted}
        \State \textbf{my-shard}  \Comment{identifier for the local shard}
        \Statex
        \Procedure{Recover}{$x$} \Comment{$x \in  \peers \cup \readfrom$} \label{alg:line:recover}
        \State $s\gets x$.shard
        \State $(w, \txidx^*) \gets$ \Call{GetStatus}{s}
        \If{$\txidx^* \le \textbf{curr-txidx}$}
        \label{proc:recoverifbranch}
        \State $\textit{success} \gets \Call{NewReader}{w, \txidx^*}$ \label{alg:line:new-reader-return-success}
        \If{\textit{success}}
        \State \readfrom $\gets \readfrom \cup \set{w}$
        \State \Return True \Comment{recovery is successful} \label{alg:line:recovery-success-1}
        \Else
        \State \Return False \Comment{recovery failed, caller should retry}
        \EndIf
        \Else
        \label{proc:recoverelsebranch}
        \For{$p\in\peers$}
        \State $\Call{Send}{p, \notifysync}$ \label{alg:line:notifysync}
        \EndFor
        \State $\Call{Synchronize}$ \label{alg:line:sync}
        \State \Return true \Comment{May run recovery for each crashed peer} \label{alg:line:recovery-success-2}
        \EndIf
        \EndProcedure
        \Statex
        \Procedure{GetStatus}{$s$} \Comment{$s$ is a shard identifier} \label{alg:line:get-status}
        \For{$w \in \text{shard } s$}
        \State \Call{Send}{$w$, \recover} \label{alg:line:send-recover}
        \EndFor
        \State $\textit{r} \gets $ receive \recoverok
        \State $\textit{replies} \gets \{r\}$
        \State $r_h \gets r$ \Comment{reply with highest txid}
        \While{$|\textit{replies}| < f+1$} \label{alg:line:get-status-wait}
        \State $\textit{r} \gets $ receive \recoverok
        \State $r_h \gets r$ if $r.\text{txid} > r_h.\text{txid}$ else $r_h$
        \State $\textit{replies} \gets \textit{replies}\cup\{r\}$
        \EndWhile

        \State \Return ($r_h.$src, $r_h.$txid)
        \EndProcedure
        \Statex
        \Procedure{NewReader}{$w, \txidx^*$} \label{alg:line:new-reader}
        \State $\newreader\gets(\txidx^*)$
        \State \Call{Send}{$w$, \newreader}
        \State $\textit{reply} \gets$ receive reply from $w$
        \If{\textit{reply} = \newreaderok}
        \State \Return true \Comment{reconfiguration success}
        \Else
        \State \Return false
        \EndIf
        \EndProcedure
    \end{algorithmic}
\end{algorithm}

In the algorithm, first the execution worker, denoted as $e$, initiates the process by establishing a view on the current execution
status of workers in shard $x$. This is achieved through a call to $(w, \txidx^*) \gets \Call{GetStatus}{x.\text{shard}}$, where $w$
represents the most up-to-date worker in a quorum of workers within $x.\text{shard}$,
having executed up to at least $\txidx^*$.

Subsequently, based on the obtained result, the execution worker $e$ takes one of two distinct actions. If $e$'s current execution
state is \textit{after} $\txidx^*$, it initiates the \Call{NewReader}{$w, \txidx^*$} operation, requesting $w$ to serve its reads that
were previously handled by $x$. Otherwise, if $e$'s current state is \textit{before} $\txidx^*$, it engages in the synchronization
procedure (\Cref{alg:synchronize}) before attempting recovery. This synchronization step becomes crucial, as there might no longer be
clusters with sufficiently old buffers for the recovering worker $e$ to proceed through the standard recovery process. In such
instances, $e$ employs the synchronization procedure to load the checkpointed state of another cluster, ensuring a seamless recovery
process in the distributed system.

\begin{algorithm}[t]
    \caption{Recovery procedure message handlers}
    \label{alg:recovery-handlers}
    \algfontsize
    \begin{algorithmic}[1]
        \Statex // Global states
        \State $\readto, \readfrom, \peers$
        \State \stabletxid \Comment{txid of locally-stored stable checkpoint}
        \State \textbf{buffer} \Comment{local store of sent \ready}
        \State \textbf{checkpoints} \Comment{snapshots of local state}
        \Statex
        \Procedure{ProcessRecover}{\recover}
        \label{proc:processrecover}
        \State $\recoverok \gets (\stabletxid)$
        \State\Call{Send}{\recover.src, \recoverok}
        \EndProcedure
        \Statex
        \Procedure{ProcessNewReader}{\newreader}
        \label{proc:processnewreader}
        \State $\textit{src}\gets$ \newreader.src
        \State $\txidx^*\gets$ \newreader.txid
        \If{$\txidx^* < \stabletxid$}
        \State\Call{Send}{\textit{src}, \recoverabort}
        \Else
        \State $\readto\gets\readto\cup\set{\textit{src}}$
        \State\Call{Send}{\textit{src}, \newreaderok}
        \For{$ r \in$ \textbf{buffer}}
        \If{\textit{r}.dst = src $\land $ \textit{r}.txid $\ge \stabletxid$ }
        \State \Call{Send}{\textit{src}, r} \Comment forward all buffered messages since \stabletxid \label{alg:line:catchup-buffered-messages}
        \EndIf
        \EndFor

        \EndIf
        \EndProcedure
        \Statex
        \Procedure{ProcessNotifySync}{\notifysync}
        \label{proc:processnotifysync}
        \State $\textit{src}\gets$\notifysync.src
        \If{\textit{src} $\in$\peers}
        \label{proc:processnotifysync1}
        \State perform  \Call{Syncrhonize}{}
        \EndIf
        \If{\textit{src} $\in$\readto}
        \label{proc:processnotifysync2}
        \State $\readto\gets\readto\setminus\set{\textit{src}}$
        \EndIf
        \If{\textit{src} $\in$\readfrom}
        \label{proc:processnotifysync3}
        \State $\readfrom\gets\readfrom\setminus\set{\textit{src}}$
        \EndIf
        \EndProcedure
        \Statex
        \Procedure{ProcessSync}{\sync}
        \label{proc:processsync}
        \State $\textit{src}\gets$ \sync.src
        \State $\txidx^*\gets$ \sync.txid
        \If{$\txidx^* < \stabletxid$}
        \State\Call{Send}{\textit{src}, \recoverabort}
        \Else
        \State \Call{Send}{\textit{src}, $\textbf{checkpoints}[\txidx^*]$}
        \EndIf
        \EndProcedure
    \end{algorithmic}
\end{algorithm}

\para{Synchronization Protocol}
\Call{Synchronize}{} (\Cref{alg:synchronize}) is called by source $e$ and brings $e$ and its peers up to date through loading the checkpointed state of another set of workers. This process on a high level works as follows.

Initially, $e$ communicates with its \readfrom and \readto nodes, issuing \notifysync messages to prompt the removal of $e$ from their respective communication sets. Additionally, $e$ notifies its \peers to cease normal operations and engage in synchronization through the \Call{Synchronize}{} procedure.

Afterwards, $e$ clears its own \readfrom and \readto sets and requests the current status of the worker $w$ in its shard by attempting to download $w$'s checkpoint. A synchronization message is sent to $w$, which responds based on whether the checkpoint at $\txidx^*$ has been deleted or not. If deleted, $e$ retries the synchronization protocol; otherwise, $w$ sends its state snapshot, and $e$ loads it using the $\Call{LoadCheckpoint}{w}$ procedure.

Finally, $e$ brodcasts a synchronization completion messages to all peers and awaits their responses. If an incoming \synccomplete has a \txidx greater than $\txidx^*$,  then $e$ retries the synchronization protocol to ensure uniformity across the entire cluster. The procedure concludes when $e$ receives \synccomplete containing \txidx from all non-suspected peers, allowing it to initiate recovery for any remaining suspected peers.

\begin{algorithm}[t]
    \caption{Synchronization procedure}
    \label{alg:synchronize}
    \algfontsize
    \begin{algorithmic}[1]
        \Procedure{Synchronize}{}
        \For{$w\in\readto\cup\readfrom$}
        \State \Call{Send}{$w$, \notifysync}
        \EndFor
        \State $\readfrom \gets \emptyset$
        \State $\readto \gets \emptyset$
        \While{true}
        \label{alg:synchronize_syncloop}
        \State $(w, \txidx^*) \gets$ \Call{GetStatus}{\textbf{my-shard}} \label{alg:line:synchronize_getstatus}
        \State $\sync\gets (\txidx^*)$
        \State \Call{Send}{$w$, \sync}
        \State $\textit{reply} \gets$ receive reply from $w$ \label{alg:recovery2syncreply}
        \If{\textit{reply} $\ne$ \recoverabort}
        \State $\Call{LoadCheckpoint}{w}$ \Comment{Load checkpoint from $w$}
        \For{$p\in\peers$}
        \State $\synccomplete\gets\txidx^*$
        \State $\Call{Send}{p, \synccomplete}$
        \EndFor
        \State $\textit{replies} \gets \{\}$ \Comment wait for \synccomplete responses
        \While{$\not\exists r. (r \in \textit{replies} \land r.\text{txid}>\txidx^*)$}
        \State $r\gets $ receive $\synccomplete$
        \If{$\forall p \in \peers\setminus\textbf{suspected}.\; \exists r. (r \in \textit{replies} \land r.\text{src} = p \land r.\text{txid} =\txidx^*)$} \label{alg:line:synchronize-syncloop-complete}
        \State \Return true
        \EndIf
        \EndWhile
        \EndIf
        \EndWhile
        \EndProcedure
    \end{algorithmic}
\end{algorithm}\textbf{}

\subsection{Proofs Modifications} \label{sec:recovery-proofs}
We show that \sysname satisfies the security properties defined in \Cref{sec:design} despite the crash-failure of $f_e$ out of $2f_e+1$ \execworkers in all shards.

\para{Assumptions}
The security of the recovery protocol relies on the following assumptions.

\begin{assumption}[Correct Majority] \label{as:correct-majority}
    At least $f_e+1$ out of $2f_e+1$ \execworkers of every shard are correct at all times.
\end{assumption}

\begin{assumption}[Eventual Synchrony] \label{as:partial-synchrony}
    The network between \execworkers is eventually synchronous~\cite{dwork1988consensus}.
\end{assumption}
It has been shown that in eventual synchrony, crash failures can eventually be perfectly detected~\cite{chandra1996}. Thus, \Cref{as:partial-synchrony} implies the following:

\begin{assumption}[Eventually Strong Failure Detector] \label{as:failure-detector}
    There exists an eventually perfect failure detector $\lozenge\mathcal{S}$ with the following properties:  (i) \emph{Strong completeness:} Every faulty process is eventually permanently suspected by every non-faulty process. 
    (ii) \emph{Eventual strong accuracy:} Eventually, correct processes are not suspected by any correct process.
\end{assumption}

\para{Serializability and determinsm proof}
We argue both the serializability and determinism of the protocol by showing that \execworkers process the same input transactions regardless of crash-faults. That is, no \execworker skips the processing of an input transaction or processes a transaction twice. Both serializability and determinsm then follow from the proofs of \Cref{sec:proofs}.

\begin{lemma}\label{lm:recovery-no-skip}
    No \execworker skips the processing of an input transaction.
\end{lemma}
\begin{proof}
    Let's assume by contradiction that a worker $w$ with state $\objectsdb$ skips the processing of the input transaction $\tx_j$. This means that (i) $w$ included in its \readfrom set a worker $w'$ with state $\objectsdb'$, where $\objectsdb'$ is the state $\objectsdb$ after the processing of the list of transactions $\tx_i,\dots,\tx_j$; and (ii) that $w$ processes a \result from $w'$ referencing transaction $\tx_j$.
    This is however a direct contradiction of check \Cref{proc:recoverifbranch} of \Cref{alg:recovery} ensuring that $w$ only includes $w'$ in its \readfrom set after its latest processed transaction \textbf{curr-txidx} is at least $\tx_j$ (that is, $\tx_j \leq \textbf{curr-txidx}$), hence a contradiction.
\end{proof}

\begin{lemma} \label{lm:recovery-no-duplicate}
    No \execworker processes the same input transaction twice.
\end{lemma}
\begin{proof}
    Let's assume by contradiction that a worker $w$ with state $\objectsdb$ processes the same input transactions $\tx_j$ twice. This means that (i) $w$ included in its \readfrom set a worker $w'$ with state $\objectsdb'$, where $\objectsdb'$ is the state $\objectsdb$ prior to the processing of the list of transactions $\tx_i,\dots,\tx_j$; and (ii) that $w$ processes a \result from $w'$ referencing transaction $\tx_j$.
    This is however impossible as $w$ could only include $w'$ in its \readfrom set after calling \Cref{alg:line:new-reader-return-success} (\Cref{alg:recovery}), and thus after $w'$ updates its state to $\objectsdb$ by processing $\tx_i,\dots,\tx_j$ at \Cref{alg:line:catchup-buffered-messages} (\Cref{alg:recovery-handlers}). As a result, $w$ could not have processed a \result from $w'$ referencing $\tx_j$ while its state is different from $\objectsdb$.
\end{proof}

\begin{lemma} \label{lm:recovery-same-input}
    \execworkers process the same set of input transactions regardless of crash-faults.
\end{lemma}
\begin{proof}
    This lemma follows from the observation that, despite crash-faults, no \execworker skips any transaction (\Cref{lm:recovery-no-skip}) nor processes any transaction twice (\Cref{lm:recovery-no-duplicate}). As a result, \execworkers process the same set of input transactions regardless of crash-faults.
\end{proof}

\para{Liveness proof}
Suppose a worker $x$ crashes. By the strong completeness property of the failure detector, a correct worker $e$ eventually detects this failure, and performs the recovery procedure (\Cref{alg:line:recover} of \Cref{alg:recovery}) to find another correct \execworker of shard $s$ to replace $x$.
We thus argue the liveness property in \Cref{lm:recovery-success} by showing that the recovery procedure presented at \Cref{alg:line:recover} of \Cref{alg:recovery} eventually succeeds; that is, it eventually exits at either \Cref{alg:line:recovery-success-1} or \Cref{alg:line:recovery-success-2}. The protocol then resumes normal operation, and the liveness of the system follows from the liveness of the normal operation protocol (\Cref{sec:proofs}).
As intermediary steps, we show that the procedures $\Call{GetStatus}{\cdot}$ (\Cref{alg:line:get-status} of \Cref{alg:recovery}), $\Call{NewReader}{\cdot}$ (\Cref{alg:line:new-reader} of \Cref{alg:recovery}), and $\Call{Synchronize}{}$ (\Cref{alg:synchronize}) eventually terminate.

\begin{lemma} \label{lm:get-status-success}
    Any call by a correct worker to $\Call{GetStatus}{\cdot}$ (\Cref{alg:line:get-status} of \Cref{alg:recovery}) eventually terminates.
\end{lemma}
\begin{proof}
    We argue this lemma by construction. Let's assume an \execworker $w$ calls $\Call{GetStatus}{s}$ on a shard $s$.
    It this sends a \recover message to all workers of shards $s$ (\Cref{alg:line:send-recover} of \Cref{alg:recovery}). By \Cref{as:partial-synchrony}, each of these workers eventually receive the messages, and correct ones reply with a \recoverok message (\Cref{proc:processrecover} of \Cref{alg:recovery-handlers}).
    By \Cref{as:correct-majority}, there are at least $f_e+1$ correct workers in shard $s$. Worker $w$ thus eventually receives at least $f_e+1$ \recoverok responses (\Cref{as:partial-synchrony}). Check \Cref{alg:line:get-status-wait} of \Cref{alg:recovery} then succeeds and ensures that $\Call{GetStatus}{s}$ returns.
\end{proof}

\begin{lemma} \label{lm:new-reader-success}
    A call $\Call{NewReader}{w, \cdot}$ (\Cref{alg:line:new-reader} of \Cref{alg:recovery}) to a correct worker $w$ eventually successfully terminates; that is, it returns \textbf{True}.
\end{lemma}
\begin{proof}
    Suppose a correct worker calls $\Call{NewReader}{w, \txidx^*}$ for a correct worker $w$. By construction, the values $(w, \txidx^*)$ are the result of the prior call to $\Call{GetStatus}{\cdot}$ (\Cref{alg:line:get-status} of \Cref{alg:recovery}). Then, given a period of synchrony where messages are delivered much quicker than checkpoint intervals (\Cref{as:partial-synchrony}), $\txidx^*$ is a valid checkpoint at $w$. As such $w$ responds to the caller's request with $\newreaderok$, and the caller successfully terminates.
\end{proof}

\begin{lemma} \label{lm:synchronize-success}
    Any call to $\Call{Synchronize}{}$ (\Cref{alg:synchronize}) eventually terminates.
\end{lemma}
\begin{proof}
    We argue this lemma by construction of \Cref{alg:synchronize}. Let $w$ be a correct worker calling \Call{Synchronize}{}. By \Cref{lm:get-status-success}, the call to $\Call{GetStatus}{\textbf{my-shard}}$ (\Cref{alg:line:synchronize_getstatus} of \Cref{alg:synchronize}) eventually returns.
    Then, $w$ eventually receives a non-\recoverabort reply ($reply \neq \recoverabort$) at \Cref{alg:recovery2syncreply} given sufficiently many executions of the loop (\Cref{alg:synchronize_syncloop}) and a period of network synchrony where messages are delivered much quicker than checkpoint intervals on other clusters (\Cref{as:partial-synchrony}).
    Worker $w$ then loads a remote snapshot, and waits for a set of \synccomplete messages from $\peers\setminus\textit{suspected}$. If $w$ receives a message with a larger $\txidx$, $w$ retries the synchronization loop at line~\ref{alg:synchronize_syncloop}. By the strong completeness property of the failure detector (\Cref{as:failure-detector}), $w$ eventually suspect all failed peers, and hence receive all responses from the $\peers\setminus\textit{suspected}$ set.
    Moreover, once messages are delivered quicker than the checkpoint intervals within clusters (\Cref{as:partial-synchrony}), all peers undergoing \Call{Synchronize}{} will synchronize to the same $\txidx^*$ after sufficient retries.
\end{proof}

\begin{lemma} \label{lm:recovery-success}
    A call to $\Call{Recover}{\cdot}$ (\Cref{alg:line:recover} of \Cref{alg:line:recover}) eventually successfully terminates. That is, it eventually exists at either \Cref{alg:line:recovery-success-1} or \Cref{alg:line:recovery-success-2}
\end{lemma}

\begin{proof}
    Consider an \execworker executing the procedure $\Call{Recover}{x}$ (\Cref{alg:line:recover} of \Cref{alg:recovery}) with $x \in  \peers \cup \readfrom$.
    The \execworker first calls $\Call{GetStatus}{x.shard}$ (\Cref{alg:line:get-status}) which is guaranteed to terminate (\Cref{lm:get-status-success}). We then have two cases: (i) the call enters the if-branch (\Cref{alg:recovery} \Cref{proc:recoverifbranch}), and (ii) the call enters the else-branch (\Cref{alg:recovery} \Cref{proc:recoverelsebranch}). We prove that the recovery procedure eventually successfully terminates in both cases.
    In the first case (i), the \execworker calls $\Call{NewReader}{w, \txidx^*}$ (\Cref{alg:line:new-reader}) which is guaranteed to eventually successfully terminate by \Cref{lm:new-reader-success}. The \execworker then adds $w$ to its \readfrom set and successfully terminates.
    In the second case (ii), \Cref{alg:line:notifysync} of \Cref{alg:recovery} ensures that every correct peer eventually performs the synchronization procedure.
    \Cref{lm:synchronize-success} then guarantees that the call to $\Call{Synchronize}{}$ (\Cref{alg:line:sync}) eventually terminates. The \execworker then successfully terminates.
\end{proof}
\section{Detailed Dynamic Objects Protocol} \label{sec:child-details}
This section completes \Cref{sec:dynamic} by providing the modifications to the algorithms of \Cref{app:algorithms} and proving the security of \sysname while supporting dynamic reads and writes.

\subsection{Algorithms Modifications} \label{sec:child-algorithms}
We specify the modifications to the algorithms of \Cref{app:algorithms} to support dynamic read and writes.

The main difference between \Cref{alg:core-child} and \Cref{alg:core} of \Cref{app:algorithms} is the removal of \Cref{alg:line:trigger-execution-freeing-read-locks}. Instead of immediately clearing the read locks after accessing the read set's objects, \Cref{alg:process-result} removes all read and write locks later.

The main change between \Cref{alg:process-ready-child} and \Cref{alg:process-ready} of Appendix~\ref{app:algorithms} is the rescheduling of $\tx$ upon discovering a dynamic object. The algorithm first calls $\Call{UppdateRWSet}{\tx\objectid'}$ \Cref{alg:line:update-rw-set-with-child} to update the read or write set of $\tx$ with the newly discovered object $\objectid'$ and then calls $\Call{RescheduleTx}{\tx,\objectid'}$ at \Cref{alg:line:reschedule-tx} to notify all concerned workers that the transaction needs to be re-scheduled for execution.

Finally, \Cref{alg:process-augmented-tx} updates the queue $\pendingdb[\objectid']$ to trigger re-execution of $\tx$ once $\objectid'$ is available.

\begin{algorithm}[t]
    \caption{Core functions (dynamic objects)}
    \label{alg:core-child}
    \algfontsize
    \begin{algorithmic}[1]
        \Function{TryTriggerExecution}{$\tx$} \label{alg:line:try-trigger-execution-child}
        \State // Check if all dependencies are already executed
        \If{$\Call{HasDependencies}{\tx}$} \Return \EndIf \label{alg:line:check-dependencies-child}
        \State
        \State // Check if all objects are present
        \State $M \gets \Call{MissingObjects}{\tx}$ \label{alg:line:check-missing-objects-child}
        \If{$M \neq \emptyset$}
        \For{$\objectid \in M$} $\missingdb[\objectid] \gets \missingdb[\objectid] \cup \tx$ \EndFor \label{alg:line:track-missing-child}
        \State \Return
        \EndIf
        \State
        \State // Send object data to a deterministically-selected \execworker
        \State $worker \gets \Call{Handler}{\tx}$ \Comment{Worker handling the most objects of $\tx$}
        \State $O \gets \{ \objectsdb[\objectid] \text{ s.t. } \objectid \in \Call{HandledObjects}{\tx} \}$ \Comment{May contain $\perp$} \label{alg:line:load-objects-child}
        \State $\ready \gets(\tx, O)$
        \State $\Call{Send}{worker, \ready}$ \label{alg:line:send-ready-child}
        \EndFunction
    \end{algorithmic}
\end{algorithm}

\begin{algorithm}[t]
    \caption{Process $\ready$ (dynamic objects)}
    \label{alg:process-ready-child}
    \algfontsize
    \begin{algorithmic}[1]
        \State $\objectsreceived \gets \{\}$ \Comment{Maps \tx to the object data it refefrernces (or $\perp$ if unavailable)}
        \Statex
        \Statex // Called by the \execworkers upon receiving a $\ready$.
        \Procedure{ProcessReady}{$\ready$}
        \State $(\tx, O) \gets \ready$
        \State $\objectsreceived[\tx] \gets \objectsreceived[\tx] \cup O$
        \If{$\len{\objectsreceived[\tx]} \neq \len{\readset{\tx}}+\len{\writeset{\tx}}$} \Return \EndIf \label{alg:line:wait-for-objects-child}
        \State
        \State $\result \gets (\tx, \emptyset, \emptyset)$ \label{alg:line:empty-result-child}
        \If{$!\Call{AbortExec}{\tx}$} \label{alg:line:can-execute-child}
        \State $r \gets \exec{\tx, \txdb[\tx]}$  \label{alg:line:execute-child}
        \If{$r = (\perp, \objectid')$}
        \State $\Call{UpdateRWSet}{\tx, \objectid'}$ \Comment{Update $\readset{\tx}$ or $\writeset{\tx}$ with $\objectid'$} \label{alg:line:update-rw-set-with-child}
        \State $\Call{RescheduleTx}{\tx, \objectid'}$ \Comment{Reschedule $\tx$ with discovered $\objectid'$} \label{alg:line:reschedule-tx}
        \State \Return
        \EndIf
        \State $(O, I) \gets r$ \Comment{$O$ to mutate and $I$ to delete}
        \For{$w \in n_e$}
        \State $O_w \gets \{ o \in O \text{ s.t. } \Call{Handler}{o}=w \}$
        \State $I_w \gets \{ \objectid \in I \text{ s.t. } \Call{Handler}{\objectid}=w \}$
        \State $\result \gets (\tx, O_w, I_w)$ \label{alg:line:send-result-child}
        \EndFor
        \EndIf
        \State $\Call{Send}{w, \result}$ \label{alg:line:send-result-2-child}
        \EndProcedure

        \Statex
        \Statex // Reschedule execution with discovered object.
        \Function{RescheduleTx}{$\tx, \objectid'$}
        \State $\augmentedtx \gets (\tx, \objectid')$
        \For{$w \in n_e$}
        \If{$\exists \objectid \in \tx \text{ s.t. } \Call{Handler}{\objectid}=w $}
        \State $\Call{Send}{w, \augmentedtx}$
        \EndIf
        \EndFor
        \EndFunction
    \end{algorithmic}
\end{algorithm}

\begin{algorithm}[t]
    \caption{Process $\augmentedtx$ (dynamic objects)}
    \label{alg:process-augmented-tx}
    \algfontsize
    \begin{algorithmic}[1]
        \Procedure{ProcessAugmentedTx}{$\augmentedtx$}
        \State $(\tx, \objectid') \gets augmentedtx$
        \If{$\objectid' \in \writeset{\tx}$}
        \State $\pendingdb[\objectid'] \gets \pendingdb[\objectid'] \cup (W, [\tx])$ \label{alg:line:schedule-write-child}
        \Else \Comment{$\objectid' \in \readset{\tx}$}
        \State $(op, T') \gets \pendingdb[\objectid'][-1]$
        \If{$op = W$}
        $\pendingdb[\objectid'] \gets \pendingdb[\objectid] \cup (R, [\tx])$  \label{alg:line:schedule-read-child}
        \Else
        $\; \pendingdb[\objectid'][-1] \gets (R, T' \cup \tx)$ \label{alg:line:schedule-parallel-read-child}
        \EndIf
        \EndIf
        \State
        \State // Try to execute the transaction
        \State $\Call{TryTriggerExecution}{\tx}$ \Comment{Defined in \Cref{alg:core}} \label{alg:line:trigger-execution-after-scheduling-child}
        \EndProcedure
    \end{algorithmic}
\end{algorithm}

\subsection{Proofs Modifications} \label{sec:child-proofs}
We specify the modifications to the proofs of \Cref{sec:proofs} to prove the serializability, determinism, and liveness (\Cref{sec:design}) of the dynamic reads and writes algorithm. The main modifications arise from the fact that \Cref{as:explicit-read-write-set} (\Cref{sec:proofs}) is not guaranteed in the dynamic reads and writes algorithm. We instead rely on \Cref{as:root-in-tx} below.

\begin{assumption}[Transaction References Root] \label{as:root-in-tx}
    If transaction $\tx$ dynamically accesses an object $\objectid'$, it explicitly references its root object $\objectid$.
\end{assumption}

The Sui MoveVM~\cite{sui-move} (used in our implementation) satisfies this assumption. As a result, this part of our design is specific to the Sui MoveVM and cannot directly generalize to other deterministic execution engines unless they implement it as well.

\para{Serializability}
We replace \Cref{lm:remove-after-access} (\Cref{sec:proofs}) with \Cref{lm:child:remove-after-process} below. The rest of the proof remains unchanged.
Intuitively, \sysname prevents the processing of conflicting transactions until all dynamic objects are discovered. This limits concurrency further than the base algorithms presented in \Cref{app:algorithms} but \Cref{sec:appendix-split-queues} shows how to alleviate this issue by indexing the queues $\pendingdb[\cdot]$ with versioned objects, that is, tuples of $(\objectid, version)$, rather than only object ids.

\begin{lemma}[Unlock after Processing] \label{lm:child:remove-after-process}
    If a transaction $\tx$ is removed from the pending queue of an object $\objectid$ then $\tx$ has already been processed (\Cref{def:processed-transaction} of \Cref{sec:proofs}).
\end{lemma}
\begin{proof}
    We argue this lemma by construction of \Cref{alg:process-result}. By definition (\Cref{def:processed-transaction}), the processing of $\tx$ terminates at \Cref{alg:process-result} \Cref{alg:line:update-stores}. However, the only way $\tx$ can be removed from $\pendingdb[\objectid]$ is by a call to $\Call{AdvanceLock}{\tx, \objectid}$. This call only occurs at one place, at \Cref{alg:line:advance-lock-call} of that same algorithm, thus after finishing the processing of $tx$.
\end{proof}

The following corollary is a direct consequence of \Cref{lm:child:remove-after-process} and facilitates the proofs presented in the rest of the section.

\begin{corollary}[Simulateous Removal] \label{lm:child:remove-simultaneous}
    A transaction $\tx$ is removed from all queues at \Cref{alg:line:advance-lock-call} of \Cref{alg:process-result}.
\end{corollary}
\begin{proof}
    We observe that the proof of \Cref{lm:child:remove-after-process} states that the only way to remove $\tx$ from a queue is by calling $\Call{AdvanceLock}{T, \objectid}$ and that this call occurs only at one place, at \Cref{alg:line:advance-lock-call} of \Cref{alg:process-result}.
\end{proof}

\para{Determinism}
The call to $exec(\tx, O)$ at \Cref{alg:line:execute} of \Cref{alg:process-ready} only completes when all objects dynamically accessed by $\tx$ are provided by the set $O$ or are specified as $\perp$. Since objects are uniquely identified by id, we need to show that all honest validators discover the same set of dynamically accessed objects.

\begin{lemma}[Consistent Dynamic Execution] \label{lm:child-exec-same-inputs}
    If a correct validator successfully calls $exec(\tx,O)$ with an dynamically accessed object $\object' \in O \text{ s.t. } \object' \neq \perp$ then no correct validators calls $exec(\tx,O)$ with $\object' = \perp$.
\end{lemma}
\begin{proof}
    Let's assume by contradiction that a correct validator $A$ calls $exec(\tx_j,O)$ (\Cref{alg:line:execute-child} of \Cref{alg:process-ready-child}) with $\object' \in O \text{ s.t. } \object' \neq \perp$ while another correct validator $B$ calls $exec(\tx_j,O)$ with $\object' = \perp$.
    This means that validator $A$ called $exec(\tx_j,O)$ after processing a previous transaction $\tx_i$ that created $\object'$, and that validator $B$ called $exec(\tx_j,O)$ before processing $\tx_i$. By \Cref{as:root-in-tx} both transactions $\tx$ and $\tx'$ conflict on the root of $\object'$, named $\object$, and \Cref{lm:queue-order} (\Cref{sec:proofs}) ensures that they are both placed in the same queue $\pendingdb[\objectid]$ (with $\objectid=\Call{Id}{\object}$). \Cref{lm:same-queues} (\Cref{sec:proofs}) ensures that both validator hold $\tx_j$ and $\tx_i$ in the same order in $\pendingdb[\objectid]$, and since validator $A$ processed $\tx_i$ before $\tx_j$, it means that both validators placed $\tx_i$ in the queue $\pendingdb[\objectid]$ before $\tx_j$.
    However \Cref{lm:child:remove-after-process} ensures that $\tx_i$ is not removed from $\pendingdb[\objectid]$ until processed and thus that validator $B$ executed $\tx_j$ despite $\tx_i$ is still in $\pendingdb[\objectid]$. However check $\Call{HasDependencies}{\tx_j}$ (\Cref{alg:line:check-dependencies-child} of \Cref{alg:core-child}) prevents  $\tx_j$ from accessing object $\object$ (since it is not at the head of the queue $\pendingdb[\objectid]$). This is a contradiction of \Cref{lm:access-then-execute} (\Cref{sec:proofs}) stating that $\tx_j$ cannot be executed before accessing $\object$.
    Since $\object'$ was chosen arbitrarily, the same reasoning applies to all objects dynamically accessed by $\tx_j$.
\end{proof}

\Cref{lm:child-exec-same-inputs} replaces the reliance on \Cref{as:explicit-read-write-set} in the proof of \Cref{th:determinism} (\Cref{sec:proofs}).

\para{Liveness}
\Cref{lm:eventual-processing} of \Cref{sec:proofs} assumes that all calls to $exec(\tx, \cdot)$ are infallible. However, supporting dynamic objects requires us to modify \Cref{alg:process-ready} as indicated in \Cref{alg:process-ready-child} and make $exec(\tx, O)$ fallible. The final \Cref{lm:child:unlock-after-process} in this paragraph proves that this change does not compromise liveness, since all dynamically accessed objects are eventually discovered, and thus all calls to $exec(\tx, \cdot)$ eventually succeed.

\begin{lemma}[Mirrored Dynamic Object Schedule] \label{lm:same-root-order}
    If a transaction $\tx_j$ is placed in a queue $\pendingdb[\objectid']$ of a dynamically accessed object $\objectid'$ after a transaction $\tx_i$, then $\tx_j$ is also placed in the queue $\pendingdb[\objectid]$ of the root object $\objectid$ after $\tx_i$.
\end{lemma}
\begin{proof}
    Let's assume by contradiction that $\tx_j$ is placed in the queue $\pendingdb[\objectid']$ of a dynamically accessed object $\objectid'$ after a transaction $\tx_i$ but before $\tx_i$ is placed in the queue $\pendingdb[\objectid]$ of the root object $\objectid$.
    By construction, $\tx_i$ can only discover $\objectid'$ upon execution (\Cref{alg:line:execute-child} of \Cref{alg:process-ready-child}). However, \Cref{lm:access-then-execute} ensures that $\tx_i$ cannot be executed before accessing $\objectid$. This means $\tx_i$ access $\objectid$ despite $\tx_j$ is already in the queue $\pendingdb[\objectid']$.
    This is however a contradiction of check $\Call{HasDependencies}{\tx_i}$ (\Cref{alg:line:check-dependencies-child} of \Cref{alg:core-child}) ensuring that $\tx_i$ is at head of $\pendingdb[\objectid]$ and thus placed in that queue before $\tx_j$.
\end{proof}

\begin{lemma}[Dynamic Access at Head of Queue] \label{lm:child:head-of-queue}
    When discovering a dynamically accessed object $\objectid'$ by executing transaction $\tx_j$ and adding $\tx_j$ to queue of $\pendingdb[\objectid']$, $\tx_j$ is at the head of the queue $\pendingdb[\objectid']$.
\end{lemma}
\begin{proof}
    Let's assume by contradiction that there exists a transaction $\tx_i$ is at the head of $\pendingdb[\objectid']$ while adding $\tx_j$ to the queue $\pendingdb[\objectid']$.
    By \Cref{as:root-in-tx}, both transactions $\tx_i$ and $\tx_j$ conflict on the root of $\objectid'$, named $\objectid$, and \Cref{lm:queue-order} (\Cref{sec:proofs}) ensures that they are both placed in the same queue $\pendingdb[\objectid]$. Given that $\tx_i$ is at the head of $\pendingdb[\objectid']$ and that \Cref{lm:child:remove-simultaneous} states that transactions are removed from all queues at the same call, $\tx_i$ is also present in the queue $\pendingdb[\objectid]$. Furthermore, since $\tx_i$ is placed in $\pendingdb[\objectid']$ before $\tx_j$, \Cref{lm:same-root-order} ensures that $\tx_i$ is also placed in the queue $\pendingdb[\objectid]$ before $\tx_j$.
    Since the discovery of the dynamic object $\objectid'$ can only occur upon executing a transaction accessing it (at \Cref{alg:line:execute-child} of \Cref{alg:process-ready-child}) and $\tx_i$ is placed in $\pendingdb[\objectid]$ before $\tx_j$, it means that \sysname executed $\tx_j$ while $\tx_i$ is still in the queue $\pendingdb[\objectid]$. This is a direct contradiction of check $\Call{HasDependencies}$ (\Cref{alg:line:check-dependencies-child} of \Cref{alg:core-child}).
\end{proof}

\begin{lemma}[Unlock after Processing] \label{lm:child:unlock-after-process}
    All objects dynamically accessed by $\tx$ are eventually discovered. That is, eventually $exec(\tx, \cdot) \neq (\perp, \cdot)$.
\end{lemma}
\begin{proof}
    \Cref{lm:child:remove-after-process} ensures that when $exec(\tx, \cdot)$ (\Cref{alg:line:execute-child} of \Cref{alg:process-ready-child}) returns $(\perp, \objectid')$, $\tx$ remains at the head of the queue $\pendingdb[\objectid]$ (for any $\objectid \in \tx$). By construction, this only happens when $\tx$ discovers a dynamic access to object $\objectid'$; $\tx$ is then added to the queue $\pendingdb[\objectid']$ (\Cref{alg:process-augmented-tx}).
    \Cref{lm:child:head-of-queue} ensures that $\tx$ is at the head of the queue $\pendingdb[\objectid']$ and thus ready for execution by referencing the newly discovered object $\objectid'$. As a result, all dynamically accessed objects are eventually discovered, and thus $exec(\tx, \cdot) \neq (\perp, \cdot)$.
\end{proof}

\subsection{Versioned Queues Scheduling}
\label{sec:appendix-split-queues}
This section shows the necessary changes to the  algorithms of \Cref{sec:child-algorithms} and data structures of \Cref{app:algorithms} to move from per-object queues to per-object-version queues.
A prerequisite for this is versioned storage of the object data itself,
that is, $\objectsdb$ should be a map $\objectsdb[\objectid, \objectversion] \rightarrow o$ instead of $\objectsdb[\objectid] \rightarrow o$,
which keeps old object versions for as long as they are referenced.
Given this, a transaction only writing (not reading) an object does not have to wait on any transaction reading the previous version.
An example of this new queueing system can be seen in \Cref{fig:split_queues_example}.
Also, the resulting dependencies between transactions can be seen in \Cref{fig:split_queues_example_happens_before}.
Without per-version queues all five transactions would have to be executed sequentially instead.

\begin{figure}[t]
    \centering
   \includegraphics[width=0.6\columnwidth]{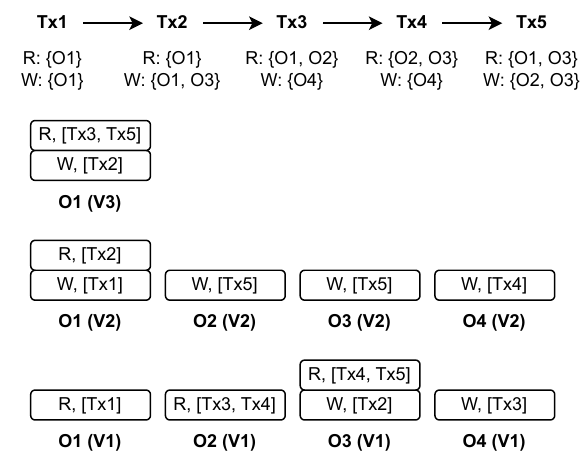}
    \caption{Example of per-object-version queues, the same transactions as in \Cref{fig:pending}.}
    \label{fig:split_queues_example}
\end{figure}

\begin{figure}[t]
    \centering\includegraphics[width=0.45\columnwidth]{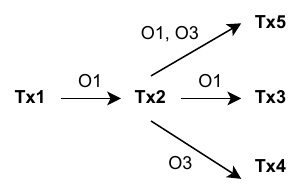}
    \caption{Example of the happens-before/waiting-on relationship resulting from the per-object-version queues of \Cref{fig:split_queues_example}.
        Edge labels indicate which object is responsible for the dependency.}
    \label{fig:split_queues_example_happens_before}
\end{figure}

One observation with these per-version queues is that each queue now only contains a single writing transaction (at the very beginning of the queue) and possibly many reading transactions following it.
This makes it straightforward to keep track of dependencies between transactions directly,
without explicitly creating the queues.
We use two maps for this:
(a) $\textsc{CurrentWriter}[\objectid] \rightarrow \txidx$ keeps track of which transaction is writing the most recent version of any object,
and (b) $\textsc{WaitingOn}[\txidx] \rightarrow [\txidx]$ keeps for each transaction a list of transactions currently writing object versions it depends on.
Additionally, a reverse mapping $\textsc{WaitedOnBy}$ can be used to enable fast deletion from $\textsc{WaitingOn}$.
Once $\textsc{WaitingOn}$ is empty for a transaction, it is ready for execution.
When enqueueing a transaction, we can check $\textsc{CurrentWriter}$ for all objects it reads to see which other transactions it needs to wait on.
This process is shown in detail in \Cref{alg:modified-process-propose} and \Cref{alg:modified-core}.

\begin{algorithm}[t!]
    \caption{Process $\propose$ (Split-Queues)}
    \label{alg:modified-process-propose}
    \algfontsize
    \begin{algorithmic}[1]
        \State $\loadedIdx \gets 0$ \Comment{All batch indices below this watermark are received}
        \State $\loaded \gets [\;]$ \Comment{Received batch indices}
        \Statex
        \Statex // Called by \execworkers upon receiving a $\propose$.
        \Procedure{ProcessPropose}{$\propose$}
        \State // Ensure we received one message per \seqworker
        \State $(\batchidx, \batchid, T) \gets \propose$
        \State $\loaded[\batchidx] \gets \loaded[\batchidx] \cup (\batchid, T)$
        \While{$\len{\loaded[\loadedIdx]} = n_s$}
        \State $(\_, T) \gets \loaded[\loadedIdx]$
        \State $\loadedIdx \gets \loadedIdx+1$
        \State
        \State // Add the objects to their pending queues
        \For{$\tx \in T$}
        \For{$\objectid \in \Call{HandledObjects}{\tx}$} \Comment{Defined in \Cref{alg:core}}
        \State $\tx' \gets \textsc{CurrentWriter}[\objectid]$
        \If{$\objectid \in \writeset{\tx}$}
        \State $\textsc{CurrentWriter}[\objectid] \gets \tx$
        \EndIf
        \If{$\objectid \in \readset{\tx}$}
        \State $\textsc{WaitingOn}[\tx'] \gets \textsc{WaitingOn}[\tx'] \cup \{\tx\}$
        \State $\textsc{WaitedOnBy}[\tx] \gets \textsc{WaitedOnBy}[\tx] \cup \{\tx'\}$
        \EndIf
        \State
        \State // Try to execute the transaction
        \State $\Call{TryTriggerExecution}{\tx}$ \Comment{Defined in \Cref{alg:core}}
        \EndFor
        \EndFor
        \EndWhile
        \EndProcedure
    \end{algorithmic}
\end{algorithm}

\begin{algorithm}[t!]
    \caption{Core functions (Split-Queues, only modified shown)}
    \label{alg:modified-core}
    \algfontsize
    \begin{algorithmic}[1]
        \State $\executedIdx \gets 0$ \Comment{All $\tx$ indices below this watermark are executed}
        \State $\executed \gets \emptyset$ \Comment{Executed transaction indices}

        \Statex
        \Function{HasDependencies}{$\tx$}
        \State \Return $\textsc{WaitingOn}[\tx] \neq \emptyset$
        \EndFunction

        \Statex
        \Function{AdvanceLock}{$\tx, \objectid$}
        \State // Cleanup the pending queue
        \State $T \gets \emptyset$
        \For{$\objectid \in \writeset{\tx}$}
        \If{$\textsc{CurrentWriter}[\objectid] = \tx$} \Comment{$\tx$ is still the most recent write}
        \State $\textsc{CurrentWriter}[\objectid] \gets \bot$
        \EndIf
        \EndFor
        \For{$\tx' \in \textsc{WaitedOnBy}[\tx]$}
        \State $\textsc{WaitingOn}[\tx'] \gets \textsc{WaitingOn}[\tx'] \setminus \{\tx\}$
        \State $\textsc{WaitedOnBy}[\tx] \gets \textsc{WaitedOnBy}[\tx] \setminus \{\tx'\}$
        \If{$\textsc{WaitingOn}[\tx'] = \emptyset$}
        \State $T \gets T \cup \{\tx'\}$
        \EndIf
        \EndFor
        \State \Return $T$
        \EndFunction
    \end{algorithmic}
\end{algorithm}
\fi

\end{document}